\documentclass[12pt,a4paper]{article}

\usepackage[english]{babel}

\usepackage[a4paper,top=2cm,bottom=2cm,left=2.5cm,right=2.5cm,marginparwidth=1.75cm]{geometry}

\usepackage[style=authoryear-comp, backend=biber, sortcites=false]{biblatex}
\usepackage{amsmath}
\allowdisplaybreaks
\usepackage{amsthm}
\usepackage{amssymb}
\usepackage{graphicx}
\usepackage[colorlinks=true, allcolors=blue]{hyperref}
\usepackage[title]{appendix}
\usepackage{mathrsfs}
\usepackage{amsfonts}
\usepackage{caption}
\usepackage{authblk}
\usepackage{csquotes}
\usepackage[multiple]{footmisc}
\usepackage[T1]{fontenc}   

\usepackage{setspace}
\usepackage{titlesec}
\titleformat{\section} 
  {\normalfont\Large\bfseries}{\thesection.}{1em}{}

\theoremstyle{plain}
\newtheorem{theorem}{Theorem}
\newtheorem{lemma}[theorem]{Lemma}
\newtheorem{proposition}[theorem]{Proposition}
\newtheorem{corollary}[theorem]{Corollary}

\theoremstyle{definition}
\newtheorem{definition}{Definition}
\newtheorem*{example}{Example}
\newtheorem{assumption}{Assumption}

\title{Existence and Calculation of Optimal Monetary Equilibria on Overlapping Generations Economies}

\author[1,2]{Leandro Lyra Braga Dognini\thanks{I am grateful to Aloisio Araujo, Yves Balasko, Luciano Irineu de Castro, Juan Pablo Gama Torres, Wilfredo Maldonado, Jaime Orrillo, Susan Schommer, and the participants of the Mathematical Economics Seminar at the Institute for Pure and Applied Mathematics (IMPA), the 20th Annual SAET Conference and the 47th Meeting of the Brazilian Econometric Society for all insightful comments on this manuscript. The paper has also benefited greatly from the comments of anonymous referees.}}
\affil[1]{\small Department of Economics, Rio de Janeiro State University}
\affil[2]{\small Legislative Advisory, Federal Senate of Brazil}
\date{May 6, 2026}  

\begin{document}
\maketitle
\begin{abstract}
\noindent A well-known feature of overlapping generations economies is that the First Welfare Theorem fails and equilibrium may be inefficient. The \textcite{Cass_1972} criterion furnishes a necessary and sufficient condition for efficiency, but it does not address the existence of efficient equilibria, and \textcite{CassOkunoZilcha_1979} provide nonexistence examples. A closely related question (known as the \textcite{Hahn_1965} problem) deals with the existence of monetary equilibria. In this paper, I provide sufficient conditions for the existence of optimal monetary equilibria on consumption-loan, non-stationary overlapping generations economies without durable, dividend-paying assets, cash-in-advance constraints, wealth-transfer mechanisms, or transaction costs. Essentially, the economy must be prone to savings. Furthermore, I develop an algorithm to find these optimal monetary equilibria as the limit of nested compact sets. These compact sets are the result of a backward calculation through equilibrium equations departing from the set of optimal monetary equilibria of well-behaved tail economies.
\end{abstract}



\section{Introduction}\label{sec1}

\textsc{Overlapping generations economies} are the workhorse for several applications of general equilibrium theory. Although the existence of equilibrium is ensured under fairly general conditions (e.g., \textcite{BalaskoCassShell_1980, Wilson_1981, OkunoZilcha_1981}), the ``double infinity of households and dated commodities'' \parencite[p. 1002]{Shell_1971} in these economies results in a striking difference from finite general equilibrium economies (e.g., \textcite{ArrowDebreu_1954}).

In overlapping generations economies, not only can equilibrium be indeterminate (e.g., \textcite{KehoeLevine_1984_1,KehoeLevine_1984_2,KehoeLevine_1985,KehoeLevine_1990,MullerWoodford_1988,KehoeLevineMas-ColellZame_1989,Burke_1990, KehoeLevineMas-ColellWoodford_1991}), but the First Welfare Theorem is also not valid; therefore, equilibrium can be inefficient (e.g., \textcite{Samuelson_1958, CassYaari_1966, Shell_1971}). This paper focuses on this second feature of these economies.

Following \textcite{Cass_1972}, the literature has derived necessary and sufficient conditions to characterize efficient equilibria (e.g., \textcite{BalaskoShell_1980,OkunoZilcha_1980, Benveniste_1986,GeanakoplosPolemarchakis_1991,RaoAiyagari_1992}) that are mainly based on the asymptotic behavior of equilibrium price sequences, in the form of
\begin{eqnarray}\label{eqEfficientEq}
    \sum^{+\infty}_{t=1}\frac{1}{\Vert p_{t}\Vert}=+\infty.
\end{eqnarray}

However, these results do not address the matter of the existence of a Pareto optimal equilibrium (i.e., if there is any equilibrium that satisfies (\ref{eqEfficientEq})), and nonexistence examples were provided by \textcite{CassOkunoZilcha_1979}.

Furthermore, in the seminal paper by \textcite{Samuelson_1958}, a Pareto optimal equilibrium not only exists, but it is achieved through the ``social contrivance of money.'' This indicates that the existence of a Pareto optimal equilibrium is closely related to a second fundamental question concerning overlapping generations economies and, more broadly, monetary theory: the existence of a monetary equilibrium (i.e., an equilibrium in which the price of money is positive). 

This matter goes back to \textcite{Hahn_1965} and \textcite{Clower_1967} critique of the Walras-Hicks-Patinkin (e.g., \textcite{Patinkin_1965}) tradition of placing \textit{money-in-the-utility-function}\footnote{See \textcite[pp. 5-8]{OstroyStarr_1990} for a complete discussion on this critique.}, and the question of the existence of a monetary equilibrium became known as the \textit{Hahn problem} (e.g., \textcite[p. 1486]{Bewley_1983}). Several solutions to this problem are found in the literature. For example, in finite horizon economies, cash-in-advance restrictions (e.g., \textcite{Clower_1967}), transaction costs (e.g., \textcite{Kurz_1974, OstroyStarr_1974}), and taxes (e.g., \textcite{Starr_1974, Starr_2003}) are all forms that ensure the existence of a monetary equilibrium.  

Infinite horizon economies are also a natural way to tackle the Hahn problem, as they rule out the zero-price case stemming from the ``terminal scrap value of money'' \parencite[p. 269]{Hayashi_1976}. Then, one possibility is to think of economies with infinite-lived households\footnote{An important property of economies with a finite number of infinite-lived households is that the First Welfare Theorem holds \parencite[pp. 104-106]{Wilson_1981} and, therefore, the matter of existence of an efficient equilibrium is straightforwardly answered once existence itself is established. This is also true in overlapping generations economies with bequests (e.g., \textcite{Barro_1974, RaoAiyagari_1989}) and in economies that possess both finite and infinite-lived households (e.g., \textcite{MullerWoodford_1988}).} (e.g., \textcite{Friedman_1969,Bewley_1972,Bewley_1980,Bewley_1983, GrandmontYounes_1972, GrandmontYounes_1973, Lucas_1980, LucasStokey_1987, BenhabibBull_1983, Levine_1989, HuShmaya_2019}) and another is to deal with overlapping generations economies. 

Therefore, there are papers in the literature on overlapping generations economies\footnote{This paper focuses only on consumption-loan deterministic models. Therefore, results regarding stochastic or production economies are not addressed.} that address these two existence questions (i.e., the existence of a Pareto optimal equilibrium and of a monetary equilibrium), either jointly or separately.

When dealing with the existence of efficient equilibria, the literature has mainly relied on the assumptions of stationary (e.g., \textcite{OkunoZilcha_1983, OkunoZilcha_1983_2,BenvenisteCass_1986,Shitovitz_1988,DucGhiglino_1998}) or nearly stationary economies (e.g., \textcite{Burke_1995}), or on the possibility of wealth transfers (e.g., \textcite{OkunoZilcha_1980,BalaskoShell_1981_1,Burke_1987}). A comprehensive discussion on this topic was written by \textcite[pp. 247-250]{Burke_1995}.

For non-stationary economies, a way of dealing with this possible lack of efficient equilibria is to assume the existence of a durable, dividend-paying asset (e.g., land), as it rules out inefficiency and reestablishes the First Welfare Theorem (e.g., \textcite[p. 773]{Mas-ColellWhinstonGreen_1995} and \textcite{ImrohorogluJoines_1999}).

When dealing with the existence of monetary equilibria, the literature has also relied on the assumptions of stationary (e.g., \textcite{Lucas_1972,Grandmont_1973,Hayashi_1976, RaoAiyagari_1987, Esteban_1994}) or nearly stationary economies (e.g., \textcite{Burke_1999}). 

For non-stationary economies, \textcite{OkunoZilcha_1982} have proven the existence of ``approximately monetary equilibria'' when there is no efficient non-monetary equilibrium. Furthermore, \textcite{Okuno_1978} and \textcite{Santos_1990} have proven the existence of monetary equilibria under cash-in-advance constraints.

Considering this overview, the main contribution of this paper is to provide sufficient conditions for the existence of an efficient monetary equilibrium in a consumption-loan, non-stationary overlapping generations economy with heterogeneous households and without durable, dividend-paying assets, cash-in-advance constraints, wealth-transfer mechanisms, or transaction costs (see Theorems \ref{theoExistenceOfParetoOptJME} and \ref{theoClassModelExistence}). The fundamental assumption is that the economy must have a natural propensity to savings, thus being a \textit{prone-to-savings economy} (see Definition \ref{defProneSavings}).

The second contribution is to show that, in consumption-loan overlapping generations economies representing \textit{long-standing economic activity}, it is adequate to explicitly model the younger period of life of generation $G_{0}$ at $t=0$ and to define the equilibrium set by: (i) imposing market clearing for all periods $t\geq1$; and (ii) limiting the aggregate demand of generation $G_{0}$ at $t=0$ according to the bounds on demographic and productivity dynamics\footnote{A related contribution is to demonstrate that, in overlapping generations economies representing long-standing economic activity, the appropriate definition of Pareto optimality to evaluate equilibrium allocations is one for which aggregate resources are only \textit{implicitly} defined by the respective allocation and, not, \textit{explicitly} carried out in the definition itself.}.

This explicit modeling of the younger period of life of generation $G_{0}$ is a fundamental departure from classical overlapping generations economies, for which the inception of the economy is assumed to occur in the first period $t=1$ (e.g., \textcite{BalaskoShell_1980, OkunoZilcha_1980, KehoeLevine_1985, RaoAiyagari_1992}). In particular, this modeling perspective allows one to endogenize the saving or borrowing decisions of households from generation $G_{0}$ without recourse to a \textit{regressum ad infinitum} setup (e.g., \textcite{Gale_1973}) and to obtain the previously mentioned existence results.

The third contribution is to highlight that efficiency conditions written in the form of (\ref{eqEfficientEq}) can be naturally related to the dynamics of real savings per capita, and in prone-to-savings economies, nonvanishing real savings per capita fully characterize efficient equilibria (see Proposition \ref{propNecAndSuffConditionParetoOpt}).

The last contribution of this paper is related to the actual calculation of Pareto optimal monetary equilibria in non-stationary overlapping generations models (and, particularly, to the calculation of efficient monetary transition paths) based on successive approximations of the reference economy by a sequence of economies that coincide with it up to a finite but arbitrarily large time period. This sequence of auxiliary economies is called a \textit{finitely replicating sequence}. One can, therefore, construct finitely replicating sequences by taking arbitrarily large sections of the reference economy and then properly choosing well-behaved \textit{tail economies}.

In this setting, Theorem \ref{theoAlgorithmFinal} provides a backward calculation algorithm to identify Pareto optimal monetary equilibria through nested compact sets obtained from a finitely replicating sequence. The algorithm is based on the fact that, by departing from the well-known set of Pareto optimal equilibria of the tail economies, we can use equilibrium equations to ``move'' backward and correctly approach at least one efficient monetary equilibrium of the reference economy. Hence, this last contribution is directly related to the literature on computational methods and overlapping generations economies (e.g., \textcite{Auerbach_1987,ImrohorogluJoines_1999, DeNardi_1999})

After this \hyperref[sec1]{Introduction}, the outline of the paper is as follows. \hyperref[sec2]{Section 2} builds the overlapping generations economy and \hyperref[sec3]{Section 3} states the main definitions and the existence theorems. \hyperref[sec4]{Section 4} then focuses on the backward calculation algorithm. Finally, a brief discussion of the results and the motivation of this paper is presented in the \hyperref[sec5]{Conclusion} and all the proofs are stated in the \hyperref[appn]{Appendix}.

\section{The overlapping generations economy}\label{sec2}

The economy $\mathcal{E}$ is a consumption-loan overlapping generations one with discrete time periods $t\in\mathbb{N}_{0}$ and $L_{t}\in\mathbb{N}$ perishable commodities in each period. Households live for two periods, the one they are born into and the next, are indexed by $h\in \mathbb{N}$ and are gathered in generations $G_{t}=\{h\in\mathbb{N}\mid h \textrm{ is born in period }t\}$, $t\geq0$, according to their period of birth. Generation $G_{t}$ has a finite number of households given by $H_{t}\in\mathbb{N}$ and $\alpha_{t}=H_{t+1}/H_{t}$ represents the demographic growth rate.

Household $h\in G_{t}$ is defined through a consumption set $X^{h}=\mathbb{R}^{L_{t}+L_{t+1}}_{+}$, a nonzero endowment $e^{h}=(e^{h}_{t},e^{h}_{t+1})\in\mathbb{R}^{L_{t}+L_{t+1}}_{+}$ and a continuous, non-decreasing, semi-strictly quasiconcave\footnote{Utility function $u(\cdot)$ is \textit{semi-strictly quasiconcave} (see \textcite[p. 87, Definition 2]{ArrowHahn_1971}), if $u(\tilde{c})\geq u(c)$ implies $u(\alpha\tilde{c}+(1-\alpha) c)\geq u(c)$, $0\leq \alpha \leq 1$, and if $u(\tilde{c})>u(c)$ implies $u(\alpha\tilde{c}+(1-\alpha) c)>u(c)$, $0< \alpha \leq1$. In particular, if $u(\cdot)$ is semi-strictly quasiconcave, then it is quasiconcave.} utility function $u^{h}:\mathbb{R}^{L_{t}+L_{t+1}}_{+}\rightarrow\mathbb{R}$ without local maxima. The average endowment of each generation is written, when convenient (e.g., Proposition \ref{propSufficientConditionProneSavings}), as $e^{t}=(e^{t}_{t},e^{t}_{t+1})=\sum_{h\in G_{t}}e^{h}/H_{t}$, $t\geq0$, and $\sum_{t\geq0}e^{t}\in\mathbb{R}^{\infty}_{++}$. I proceed with the following assumption.

\begin{assumption}[Bounded demographic growth rates and endowments]\label{assBoundsPopAndEndownments}
There are $0<\alpha_{\min}<\alpha_{\max}$ and $e_{\max}>0$ such that $\alpha_{t}\in[\alpha_{\min},\alpha_{\max}]$, $t\geq0$, and $\Vert e^{h}\Vert_{\infty}=\max_{1\leq i\leq  L_{t}+L_{t+1}}\vert e^{h}_{i}\vert\leq e_{\max}$, $h\in \mathbb{N}$.
\end{assumption}

Notice that Assumption \ref{assBoundsPopAndEndownments} limits only population growth rates and not the population itself, and the uniform upper bound over endowments can be taken to be arbitrarily large, so that every reasonable demographic and productivity dynamic can be represented\footnote{Although the economy is a consumption-loan one, the endowments can be seen as a measure of households productivity after an inelastic labor supply.}. The equilibrium existence result of \textcite[p. 119, Theorem 5]{ArrowHahn_1971} leads to the next assumption.

\begin{assumption}\label{assResourceRelated}
For $j\geq 0$, all households $h\in \bigcup^{j}_{t=0}G_{t}$ are indirectly resource related\footnote{Following \textcite[117-118]{ArrowHahn_1971}, given $j\geq0$, household $h^{\prime}\in\bigcup^{j}_{t=0}G_{t}$ is \textit{resource related} to household $h^{\prime\prime}\in\bigcup^{j}_{t=0}G_{t}$ if, for every allocation $\{c^{h}\}_{h\in\bigcup^{j}_{t=0}G_{t}}$, $c^{h}\in\mathbb{R}^{L_{t}+L_{t+1}}_{+}$, $h\in G_{t}$, $0\leq t\leq  j$, with $\sum_{0\leq t\leq  j}\sum_{h\in G_{t}}c^{h}\leq \sum_{0\leq t\leq  j}\sum_{h\in G_{t}}e^{h}$, there exists an allocation $\{\tilde{c}^{\, h}\}_{h\in\bigcup^{j}_{t=0}G_{t}}$ and $y\in\mathbb{R}^{\sum^{j+1}_{i=0}L_{i}}_{+}$ such that $\sum_{0\leq t\leq  j}\sum_{h\in G_{t}}\tilde{c}^{\, h}\leq\sum_{0\leq t\leq  j}\sum_{h\in G_{t}}e^{h} +y$,
\begin{eqnarray*}
        u^{h}(\tilde{c}^{\, h})&\geq&u^{h}(c^{h}),\textrm{ for all }h\in\bigcup^{j}_{t=0}G_{t},\\
        u^{h^{\prime\prime}}(\tilde{c}^{\, h^{\prime\prime}})&>&u^{h^{\prime\prime}}(c^{h^{\prime\prime}}),
\end{eqnarray*}
and
\begin{eqnarray*}
    y_{ti}>0\textrm{ only if } e^{h^{\prime}}_{ti}>0,
\end{eqnarray*}
for $i\in\{1,\ldots, L_{t}\}$, $0\leq t\leq j+1$. Household $h^{\prime}$ is \textit{indirectly resource related} to household $h^{\prime\prime}$ if there is a sequence of households $\{h_{m}\}_{0\leq m\leq n}$, $n\geq1$, with $h_{0}=h^{\prime}$, $h_{n}=h^{\prime\prime}$ and $h_{m}$ resource related to $h_{m+1}$, $0\leq m\leq n-1$. 
}.
\end{assumption}

Household $h\in G_{t}$, $t\geq0$, has a Walrasian demand function $x^{h}:\mathbb{R}^{L_{t}+L_{t+1}}_{++}\rightarrow \mathbb{R}^{L_{t}+L_{t+1}}_{+}$, $x^{h}(p_{t},p_{t+1})=(x^{h}_{t}(p_{t},p_{t+1}),x^{h}_{t+1}(p_{t},p_{t+1}))$, $(p_{t},p_{t+1})\in\mathbb{R}^{L_{t}+L_{t+1}}_{++}$, given by 
\begin{eqnarray*}
    x^{h}(p_{t},p_{t+1})=\textrm{argmax}_{c\in\mathbb{R}^{L_{t}+L_{t+1}}_{+}}& u^{h} (c) \\
    \textrm{ s.t. }& (p_{t},p_{t+1})\cdot (c-e^{h})\leq0.
\end{eqnarray*}
The excess demand function in period $t\geq1$, $z_{t}:\mathbb{R}^{L_{t-1}+L_{t}+L_{t+1}}_{++}\rightarrow \mathbb{R}^{L_{t}}$, is 
\begin{eqnarray*}
    z_{t}(p_{t-1},p_{t},p_{t+1})=\sum_{h\in G_{t-1}}(x^{h}_{t}(p_{t-1},p_{t})-e^{h}_{t})+\sum_{h\in G_{t}}(x^{h}_{t}(p_{t},p_{t+1})-e^{h}_{t}),
\end{eqnarray*}
and joint excess demand until period $t\geq1$, $Z_{t}:\mathbb{R}^{\sum^{t+1}_{i=0}L_{i}}_{++}\rightarrow\mathbb{R}^{\sum^{t}_{i=1}L_{i}}$, is
\begin{eqnarray*}
    Z_{t}(p_{0},\ldots,p_{t+1})=(z_{1}(p_{0},p_{1},p_{2}),\ldots,z_{t}(p_{t-1},p_{t},p_{t+1})).
\end{eqnarray*}
It is also convenient to define $\mathcal{Z}_{t}\subseteq \mathbb{R}^{\sum^{t+1}_{i=0}L_{i}}_{++}$, $t\geq1$, as 
\begin{eqnarray*}
    \mathcal{Z}_{t}=\{(p_{0},\ldots,p_{t+1})\in\mathbb{R}^{\sum^{t+1}_{i=0}L_{i}}_{++} \mid p_{01}=1, Z_{t}(p_{0},\ldots,p_{t+1})=0\}.
\end{eqnarray*}
For $c\in\mathbb{R}^{L}$, $L\geq1$, $\Vert c\Vert_{\infty}=\max_{1\leq i\leq L}\vert c_{i}\vert$ and $\Vert c\Vert=\sum^{L}_{i=1}\vert c_{i}\vert$ and, as a topological space, $\mathbb{R}^{\infty}$ is endowed with the product topology. Also, when convenient (e.g., Definition \ref{defParetoOptimalEqSet}), I consider vectors in $\mathbb{R}^{L_{t}+L_{t+1}}$, $t\geq0$, as elements of $\mathbb{R}^{\infty}$, with the appropriate embedding.

Before stating the third assumption, let $\beta=\max\, \{1+\alpha^{-1}_{\min},1+\alpha_{\max}\}$ and $\mathcal{B}_{t}(\sigma)=\{p\in\mathbb{R}^{L_{t}+L_{t+1}}_{++}\mid \sigma \leq p_{i}\leq \sigma^{-1}, 1\leq i \leq L_{t}+L_{t+1}\}$, for $\sigma\in(0,1)$, $t\geq 0$. It follows directly from its definition that the ``box-set'' $\mathcal{B}_{t}(\sigma)\subset \mathbb{R}^{L_{t}+L_{t+1}}_{++}$ is compact. 
\begin{assumption}[Unbounded aggregate demand at border prices]\label{assBoundsPrices}
For all $t\geq0$, there is $\sigma_{t}\in(0,1)$ such that
\begin{eqnarray*}
    \biggr(\frac{p_{t}}{p_{t1}},\frac{p_{t+1}}{p_{t1}}\biggr)\notin \mathcal{B}_{t}(\sigma_{t})\implies \biggr\Vert \sum_{h\in G_{t}} \frac{x^{h}(p_{t},p_{t+1})}{H_{t}}\biggr\Vert_{\infty} > \beta e_{\max},
\end{eqnarray*}
for $(p_{t},p_{t+1})\in \mathbb{R}^{L_{t}+L_{t+1}}_{++}$.
\end{assumption}

Assumption \ref{assBoundsPrices} states that if the prices faced by generation $G_{t}$ move sufficiently close to the border of the corresponding simplex, then its aggregate demand will extrapolate the existing resources on some of the commodities of periods $t$ or $t+1$, considering all possible demographic and productivity trends compatible with Assumption \ref{assBoundsPopAndEndownments}. This kind of boundary behavior of aggregate demand functions is common in the literature (e.g., \textcite[pp. 581-582]{Mas-ColellWhinstonGreen_1995} and \textcite[p. 31]{ArrowHahn_1971}) and does not impose a significant constraint on the model.

It is also convenient to define the following compact sets $\mathcal{K}_{t}\subset\mathbb{R}^{L_{t}}_{++}$, $t\geq0$, and $\mathcal{K}\subset\mathbb{R}^{\infty}_{++}$,
\begin{eqnarray*}
    \mathcal{K}_{0}&=&\{p\in\mathbb{R}^{L_{0}}_{++}\mid \sigma_{0}\leq p_{i}\leq \sigma_{0}^{-1}, 1\leq i\leq L_{0}\}\\
    \mathcal{K}_{t}&=&\biggr\{p\in\mathbb{R}^{L_{t}}_{++}\mid \prod^{t-1}_{j=0}\sigma_{j}\leq p_{i}\leq \biggr(\prod^{t-1}_{j=0}\sigma_{j}\biggr)^{-1}, 1\leq i\leq L_{t}\biggr\},
\end{eqnarray*}
for $t\geq1$, and $\mathcal{K}=\prod^{+\infty}_{t=0}\mathcal{K}_{t}$.

There are still two remarks to be made about the economy $\mathcal{E}$. First, since the number of commodities in each period $L_{t}\in\mathbb{N}$ is arbitrary, utility functions $u^{h}(\cdot)$ are only non-decreasing and semi-strictly quasiconcave, and endowments $e^{h}\in\mathbb{R}^{L_{t}+L_{t+1}}_{+}$ are not necessarily strictly positive, the fact that every household lives for two periods does not impose any real restriction on the model. This is due to an algorithm conceived by \textcite[319-321]{BalaskoCassShell_1980} for relabeling periods and commodities from any overlapping generations economy with finite-lived households in order to pose it as a two-period model. 

The last remark is that, unlike the classical overlapping generations economies (e.g., \textcite{BalaskoShell_1980, OkunoZilcha_1980, KehoeLevine_1985,RaoAiyagari_1992}), I label time periods so that there is no generation that lives for a single period in $t=1$, which is classically referred to as the ``inception of the economy.''

This labeling, along with the definition of the set of equilibria in \hyperref[sec3]{Section 3}, is a distinguishing feature of this model since my aim is to regard the economy $\mathcal{E}$ as the continuation of all possible overlaps that occurred before generation $G_{0}$ was born, due to long-standing economic activity. In particular, it allows us to treat, within a common framework, economies with and without money (e.g., \textcite{BalaskoShell_1981_1}, \textcite{BalaskoShell_1980}), and in the former case, monetary and barter equilibria (e.g., \textcite{CassOkunoZilcha_1979}). This last remark will be extensively discussed in Sections \hyperref[sec3]{3} and \hyperref[sec4]{4}.

\section{The set of equilibria and the existence theorems}\label{sec3}

\begin{definition}\label{defSetEquilibria}
The \textit{set of equilibria} $\mathcal{H}\subset\mathbb{R}^{\infty}_{++}$ of economy $\mathcal{E}$ is
\begin{eqnarray*}
\mathcal{H}=\{p\in\mathbb{R}^{\infty}_{++}\mid p_{01}=1,(p_{0},p_{1})\in\mathcal{B}_{0}(\sigma_{0})\textrm{ and } z_{t}(p_{t-1},p_{t},p_{t+1})=0, t\geq 1\}.
\end{eqnarray*}
\end{definition}

Definition \ref{defSetEquilibria} states that, due to the homogeneity of demand, all price sequences $p=(p_{0},p_{1},\ldots)\in\mathbb{R}^{\infty}_{++}$ in the set of equilibria $\mathcal{H}$ are defined in terms of the first good in period $t=0$. In addition, Definition \ref{defSetEquilibria} imposes market clearing in all periods $t\geq1$. 

However, in period $t=0$, an old generation is ``lacking,'' and Definition \ref{defSetEquilibria} only requires the aggregate demand of generation $G_{0}$ to be bounded according to Assumptions \ref{assBoundsPopAndEndownments} and \ref{assBoundsPrices}. Therefore, there are no market clearing equations to be solved in period $t=0$, but all households of generation $G_{0}$ must have their demands defined through common prices $(p_{0},p_{1})\in\mathbb{R}^{L_{0}+L_{1}}_{++}$ that are compatible with Assumptions \ref{assBoundsPopAndEndownments} and \ref{assBoundsPrices}.

One must pay close attention to this feature of Definition \ref{defSetEquilibria}, as it represents a fundamental departure from the classical equilibrium definition in overlapping generations economies (e.g., \textcite{CassOkunoZilcha_1979}, \textcite{BalaskoShell_1980,BalaskoShell_1981_1}, \textcite{OkunoZilcha_1980}, \textcite{KehoeLevine_1985} and \textcite{RaoAiyagari_1992}).

In the classical model, the economy is indexed by time periods $t\geq1$ and the generation born in $t=0$ is ``present at the inception of the economy (in which case, they live out the balance of their lives during period 1)'' \parencite[p. 283]{BalaskoShell_1980}, with markets clearing in all periods. To illustrate the difference, consider classical economies in which all generations born in or after period $t=1$ are identical to those of $\mathcal{E}$. 

These classical economies can be divided into two types. In the first, there is no fiat or outside money (e.g., \textcite{BalaskoShell_1980}), meaning that there is no durable good that serves only as a store of value. Let $\mathcal{E}^{1}$ be one of these economies. The utility maximization problem of $h\in G^{1}_{0}$ is written as 
\begin{eqnarray*}
    x^{h}_{1}(p_{1})=\textrm{argmax}_{c\in\mathbb{R}^{L_{1}}_{+}}& u^{h} (c) \\
    \textrm{ s.t. }& p_{1}\cdot (c-e^{h}_{1})\leq0,
\end{eqnarray*}
for $p_{1}\in\mathbb{R}^{L_{1}}_{++}$ \parencite[p. 285]{BalaskoShell_1980}, and all other generations $G^{1}_{t}$, $t\geq1$, behave exactly as those from $\mathcal{E}$. If $(p_{1},p_{2},\ldots)\in\mathbb{R}^{\infty}_{++}$ satisfies the market clearing equations in all periods $t\geq1$, then it is an equilibrium of $\mathcal{E}^{1}$.

The second type of classical economies assumes that fiat money exists. Let $\mathcal{E}^{2}$ be one of these economies. In this case, household $h\in G^{2}_{0}$ has an endowment of fiat money $m^{h}\geq0$, with $\sum _{h\in G^{2}_{0}}m^{h}=1$ \parencite[p. 45]{CassOkunoZilcha_1979}. The utility maximization problem of $h\in G^{2}_{0}$ is written as 
\begin{eqnarray*}
    x^{h}_{1}(p_{m},p_{1})=\textrm{argmax}_{c\in\mathbb{R}^{L_{1}}_{+}}& u^{h} (c) \\
    \textrm{ s.t. }& p_{1}\cdot (c-e^{h}_{1})\leq p_{m}m^{h},
\end{eqnarray*}
with $(p_{m},p_{1})\in\mathbb{R}_{+}\times\mathbb{R}^{L_{1}}_{++}$. All other generations $G^{2}_{t}$, $t\geq1$, behave exactly as those from $\mathcal{E}$. If $(p_{m},p_{1},p_{2},\ldots)\in\mathbb{R}_{+}\times\mathbb{R}^{\infty}_{++}$ satisfies the market clearing equations in all periods $t\geq1$, then it is an equilibrium. If $p_{m}=0$, it is called a $\textit{barter equilibrium}$ and, if $p_{m}>0$, it is called a $\textit{monetary equilibrium}$ \parencite[46]{CassOkunoZilcha_1979}\footnote{Notice that a stationary economy $\mathcal{E}^{2}$ with a stationary monetary equilibrium is what \textcite[20]{Gale_1973} termed a \textit{Samuelson model}.}. 

There is an equivalence between the set of equilibria of $\mathcal{E}^{1}$  
and the set of barter equilibria of $\mathcal{E}^{2}$, i.e., by defining 
\begin{eqnarray*}
    \mathcal{H}^{1}=\{(0,p_{1},p_{2},\ldots)\in\mathbb{R}_{+}\times\mathbb{R}^{\infty}_{++}\mid (p_{1},p_{2},\ldots)\textrm{ is an equilibrium of }\mathcal{E}^{1}\}\\
    \mathcal{H}^{2}=\{(p_{m},p_{1},p_{2},\ldots)\in\mathbb{R}_{+}\times\mathbb{R}^{\infty}_{++}\mid (p_{m},p_{1},p_{2},\ldots)\textrm{ is an equilibrium of }\mathcal{E}^{2}\}\\
    \mathcal{H}^{2}_{b}=\{(0,p_{1},p_{2},\ldots)\in\mathbb{R}_{+}\times\mathbb{R}^{\infty}_{++}\mid (0,p_{1},p_{2},\ldots)\textrm{ is a barter equilibrium of }\mathcal{E}^{2}\},
\end{eqnarray*}
we have $\mathcal{H}^{1}=\mathcal{H}^{2}_{b}\subseteq \mathcal{H}^{2}$. Therefore, adding money to the model possibly, though not necessarily, enlarges the set of equilibria\footnote{See \textcite{CassOkunoZilcha_1979} and \textcite{OkunoZilcha_1982} for several counterexamples on conjectures dealing with the existence and optimality of barter and monetary equilibria in classical overlapping generations economies.}.

However, from a modeling perspective, $\mathcal{E}^{1}$ and $\mathcal{E}^{2}$ have some undesirable characteristics. First, unless one is actually modeling the inception of economic activity, what happened in the ``past overlaps'' cannot be fully ignored by the model, since it should ``inherit,'' in a sense that will later be made precise, the consequences of past economic activity.

Second, the economy $\mathcal{E}^{1}$ implicitly assumes that the creation of any form of fiat money is impossible and, therefore, that only barter equilibria exist. This is a strong assumption. On the other hand, the economy $\mathcal{E}^{2}$ allows both barter and monetary equilibria to exist and does not exclude a priori the creation of fiat money. Considering that $\mathcal{H}^{1}\subseteq\mathcal{H}^{2}$ and that monetary equilibria can Pareto dominate barter equilibria (e.g., \textcite{Samuelson_1958}), it is reasonable to assume that the creation of fiat money is a possibility and, in this sense, $\mathcal{E}^{2}$ is a more adequate model.

One could correctly point out that past overlaps are not necessarily ignored by $\mathcal{E}^{2}$, since $m^{h}\geq0$, $h\in G^{2}_{0}$, can be interpreted not as fiat money created, but as fiat money previously acquired; although in this type of model the acquisition process itself is exogenous. Furthermore, this interpretation carries the implicit assumption that all households of generation $G^{2}_{0}$ never consumed, in their previous period of life, more than their endowment could afford (i.e., there are no borrowers in $G^{2}_{0}$). 

Nonetheless, since households are heterogeneous, it is possible that intragenerational trade has occurred in period $t=0$ and some households of $G^{2}_{0}$ have become borrowers of their own cohort (i.e., along with fiat or outside money, generation $G^{2}_{0}$ also carries IOU or inside money). In order to allow this sort of economic activity in the model, one could simply let $m^{h}\in \mathbb{R}$ and interpret it not only as fiat money but also as private debt obligations standing in period $t=1$. Let $\mathcal{E}^{3}$ be one of these economies in which $m^{h}\in\mathbb{R}$, $h\in G^{3}_{0}$.

Definition \ref{defSetEquilibria} closely relates economy $\mathcal{E}$ and this third type due to the following reasoning. Looking back at the economy $\mathcal{E}$, let $(p_{0},p_{1},\ldots)\in\mathcal{H}$. Then, for all $h\in G_{0}$, let $h^{\prime}$ be a household of $G^{3}_{0}$ with $u^{h^{\prime}}:\mathbb{R}^{L_{1}}_{+}\rightarrow\mathbb{R}$ given by $u^{h^{\prime}}(c)=u^{h}(x^{h}_{0}(p_{0},p_{1}),c)$, $e^{h^{\prime}}_{1}=e^{h}_{1}$, and $m^{h^{\prime}}=p_{0}\cdot(e^{h}_{0}-x^{h}_{0}(p_{0},p_{1}))\in\mathbb{R}$. Then, let $\{u^{h^{\prime}}(\cdot),e^{h^{\prime}}_{1},m^{h^{\prime}}\}_{h\in G_{0}}$ define the generation $G^{3}_{0}$ of economy $\mathcal{E}^{3}$, which has all generations born in or after period $t=1$ identical to those of $\mathcal{E}$. The definitions of $x^{h^{\prime}}(\cdot)$, $e^{h^{\prime}}_{1}$ and $m^{h^{\prime}}$ imply that
\begin{eqnarray*}
    x^{h^{\prime}}_{1}(1,p_{1})=x^{h}_{1}(p_{0},p_{1}),
\end{eqnarray*}
and, therefore, market clearing equations in periods $t\geq1$ for $\mathcal{E}$ imply $(1,p_{1},p_{2},\ldots)\in\mathcal{H}^{3}$, with
\begin{eqnarray*}
    \mathcal{H}^{3}=\{(p_{m},p_{1},p_{2},\ldots)\in\mathbb{R}_{+}\times\mathbb{R}^{\infty}_{++}\mid (p_{m},p_{1},p_{2},\ldots)\textrm{ is an equilibrium of }\mathcal{E}^{3}\}.
\end{eqnarray*}
If $m^{h}=0$, $h\in G^{3}_{0}$, then $(1,p_{1},p_{2},\ldots)\in\mathcal{H}^{3}$ is equivalent to an equilibrium of an economy without fiat money $\mathcal{E}^{1}$ or to a barter equilibrium of an economy with fiat money $\mathcal{E}^{2}$. If $m^{h}\geq0$, $h\in G^{3}_{0}$, and $\sum_{h\in G^{3}_{0}}m^{h}>0$, then $(1,p_{1},p_{2},\ldots)\in\mathcal{H}^{3}$ is equivalent to a monetary equilibrium of an economy $\mathcal{E}^{2}$. 

However, several other configurations do not match the classical economies. For example, if $m^{h}<0$, for some $h\in G^{3}_{0}$, and $\sum_{h\in G^{3}_{0}}m^{h}>0$, then $(1,p_{1},p_{2},\ldots)\in\mathcal{H}^{3}$ defines an equilibrium that can be seen as an economy in which fiat money and IOU are simultaneously held by households from generation $G^{3}_{0}$. The fact that $\mathcal{H}$ gathers all these different configurations, which include the classical ones, is the fundamental property of $\mathcal{E}$.

This reasoning also reveals that every equilibrium $(p_{0},p_{1},\ldots)\in\mathcal{H}$ defines an equilibrium of an economy $\mathcal{E}^{3}$, $(1,p_{1},p_{2},\ldots)\in\mathcal{H}^{3}$, with $\mathcal{E}^{3}$ inheriting utility functions $u^{h^{\prime}}(\cdot)$ and the values of fiat money or private debt obligations  $m^{h^{\prime}}\in\mathbb{R}$, $h\in G^{3}_{0}$, from past economic activity (i.e., past overlaps). Although all the past overlaps are not explicitly modeled, Definition \ref{defSetEquilibria} requires that the consumption bundles of generation $G_{0}$ must be derived from competitive prices $(p_{0},p_{1})$ that are compatible with Assumptions \ref{assBoundsPopAndEndownments} and \ref{assBoundsPrices}.

But if economies $\mathcal{E}$ and $\mathcal{E}^{3}$ are related, why not work directly with $\mathcal{E}^{3}$ instead of $\mathcal{E}$? First, because $\mathcal{E}^{3}$ was derived from $\mathcal{E}$ – it is a sort of ``chopped'' part of it. Second, because $\mathcal{E}$ explicitly models the fundamentals that allow households, including those born in period $t=0$, to become borrowers or savers. It can be seen as a way to make the saving or borrowing decisions of all generations endogenous, without recourse to a \textit{regressum ad infinitum} setup (e.g., \textcite{Gale_1973}).

Third, because the fundamental question that is being answered when dealing with the economy $\mathcal{E}$ is: What equilibrium prices are consistent with long-standing economic activity? Or, stated differently, if economic activity started long ago and, therefore, we are not at the inception of the economy, what prices are compatible with equilibrium under perfect foresight?

One could still argue that it is, nevertheless, possible to define an economy $\mathcal{E}^{3}$ by arbitrarily choosing the number of households born in period $t=0$, along with their utility functions, endowments, and fiat money or private debt obligations $m^{h}\in\mathbb{R}$. So why not proceed in this manner and study the set of equilibria of this economy? The main reason is technical. If $m^{h}<0$, for some $h\in G^{3}_{0}$, the budget set becomes empty at certain prices, and the existence of equilibrium is no longer guaranteed. If $m^{h}\geq0$, for all $h\in G^{3}_{0}$, then this is again the first or the second type of classical economies (i.e., $\mathcal{E}^{1}$ or $\mathcal{E}^{2}$).

Why, then, did the classical models not adopt the setup of $\mathcal{E}$ and the definition of the equilibrium set $\mathcal{H}$? I believe the answer to this question involves at least two points. First, the classical models were not looking for results like those in this paper (e.g., Theorem \ref{theoAlgorithmFinal}), in which one needs to take conveniently chosen overlapping generations economies (which will later be called \textit{tail economies} by Definition \ref{defFiniteRep}) and use them to replace all generations born after a given period in a reference economy. For this sort of replacement to be possible, the older generation that is also alive at $t=0$ in this conveniently chosen tail economy cannot be explicitly modeled.

Second, because the Pareto optimal definition for consumption-loan economies commonly assumes some concept of \textit{feasibility} (e.g., \textcite[798-800]{OkunoZilcha_1980}), which presumes fixed aggregate resources for the economy. However, when looking at the allocations implied by $\mathcal{H}$, aggregate resources are only fixed for $t\geq1$, in which case they are given by $\sum_{h\in G_{t-1}}e^{h}_{t}+\sum_{h\in G_{t}}e^{h}_{t}$. 

In order to have fixed aggregate resources in $t=0$ and, consequently, to adopt a concept of feasibility, one could impose a market clearing equation in $t=0$, $\sum_{h\in G_{0}} x^{h}(p_{0},p_{1})=\sum_{h\in G_{0}}e^{h}_{0}$. But then we would obtain precisely the classical model of an economy without fiat money $\mathcal{E}^{1}$ (for the coincidence to be complete, one would also need to relabel periods $t=0$ and $t=1$ as the initial period).

Notice, however, that the concept of Pareto optimality can be stated by looking solely at the allocation (e.g., \textcite[286]{BalaskoShell_1980}), without any explicit reference to matters of feasibility. In this case, aggregate resources are implicitly defined by the allocation itself, allowing us to use the Pareto optimality concept to characterize elements of $\mathcal{H}$. This leads to the following definition.

\begin{definition}\label{defParetoOptimalEqSet}
The \textit{subset of Pareto optimal equilibria} $\mathcal{H}^{PO}\subseteq\mathcal{H}$ is
\begin{eqnarray*}
    \mathcal{H}^{PO}=\{p\in\mathcal{H}\mid \nexists \{y^{h}\}_{h\geq1}, y^{h}\in\mathbb{R}^{L_{t}+L_{t+1}}_{+}, h\in G_{t},t\geq0, \textrm{ s.t. } \sum_{t\geq0}\sum_{h\in G_{t}}(y^{h}-x^{h}(p_{t},p_{t+1}))=0\\
    \textrm{and }u^{h}(y^{h})\geq u^{h}(x^{h}(p_{t},p_{t+1})), h\in G_{t}, t\geq0, \textrm{ with at least one strict inequality}\}.
\end{eqnarray*}
\end{definition}

Definition \ref{defParetoOptimalEqSet} identifies the subset of $\mathcal{H}$ for which there is no possibility of redistribution of aggregate resources leading to a Pareto improvement. Furthermore, due to Definition \ref{defSetEquilibria}, aggregate resources implicitly defined by elements of $\mathcal{H}$ can vary only in period $t=0$.

There is a fundamental remark about the subset of Pareto optimal equilibria $\mathcal{H}^{PO}$. Notice that for $p\in\mathcal{H}$, it is always possible to build another overlapping generations economy $\mathcal{E}_{p}$, with the same households from $\mathcal{E}$, but with endowments given by $e^{h}_{p}=x^{h}(p_{t},p_{t+1})$, $h\in G_{t}$, $t\geq0$. In this economy $\mathcal{E}_{p}$, the new demand of household $h\in G_{t}$, $t\geq0$, is given by
\begin{eqnarray*}
    x^{h}_{p}(q_{t},q_{t+1})=\textrm{argmax}_{c\in\mathbb{R}^{L_{t}+L_{t+1}}_{+}}& u^{h} (c) \\
    \textrm{ s.t. }& (q_{t},q_{t+1})\cdot (c-e^{h}_{p})\leq0.
\end{eqnarray*}
for $(q_{t},q_{t+1})\in\mathbb{R}^{L_{t}+L_{t+1}}_{++}$. We straightforwardly see that $ x^{h}_{p}(p_{t},p_{t+1})=e^{h}_{p}$. In particular, $\sum_{t\geq0}\sum_{h\in G_{t}}x^{h}_{p}(p_{t},p_{t+1})-e^{h}_{p}=0$ and, therefore, $p\in\mathbb{R}^{\infty}_{++}$ is an equilibrium of the economy $\mathcal{E}_{p}$ (actually, an \textit{autarkic} or \textit{no-trade} equilibrium), with market clearing equations satisfied at all periods, including $t=0$ (i.e., $\sum_{h\in G_{0}}x^{h}_{p0}(p_{0},p_{1})-e^{h}_{p0}=0$). 

By relabeling periods $t=0$ and $t=1$ as the initial period, notice that $\mathcal{E}_{p}$ is actually the classical setup of an overlapping generations economy without fiat money $\mathcal{E}^{1}$. Also, $p\in\mathcal{H}^{PO}$ if, and only if, $p$ is a Pareto optimal equilibrium of $\mathcal{E}_{p}$.

Then, if the \textcite{Cass_1972} criterion can be applied to characterize the set of Pareto optimal equilibria of $\mathcal{E}_{p}$, the following equivalence holds
\begin{eqnarray*}
    p\in\mathcal{H}^{PO}\iff \sum^{+\infty}_{t=0}\frac{1}{H_{t}\Vert p_{t}\Vert}=+\infty,
\end{eqnarray*}
and, if this reasoning can be applied for all $p\in\mathcal{H}$, then
\begin{eqnarray}\label{AlternativeDefParetoOptimal}
    \mathcal{H}^{PO}=\biggr\{(p_{0},p_{1},\ldots)\in\mathcal{H}\mid \sum^{+\infty}_{t=0}\frac{1}{H_{t}\Vert p_{t}\Vert}=+\infty\biggr\}.
\end{eqnarray}
This leads to the following assumption.
\begin{assumption}[\textcite{Cass_1972} criterion]\label{assCassCriterion}
    Let $p\in\mathcal{H}$. Then, $p\in\mathcal{H}^{PO}$ if, and only if, 
    \begin{eqnarray*}
         \sum^{+\infty}_{t=0}\frac{1}{H_{t}\Vert p_{t}\Vert}=+\infty.
    \end{eqnarray*}
\end{assumption}

As described in the \hyperref[sec1]{Introduction}, following \textcite{Cass_1972}, \textcite{BalaskoShell_1980,OkunoZilcha_1980,Benveniste_1986,GeanakoplosPolemarchakis_1991} derived conditions under which Assumption \ref{assCassCriterion} can be sustained\footnote{All the existence results in Section \ref{sec3} are, in fact, results on the existence of equilibrium price sequences that satisfy the Cass criterion (see the proof of Theorem \ref{theoClassModelExistence}, where this point is raised). Therefore, Assumption \ref{assCassCriterion} could require only the sufficiency of the Cass criterion for optimality (e.g., when the conditions of Theorem 3A from \textcite[p. 803]{OkunoZilcha_1980} are satisfied), and this would still ensure the existence of \textit{optimal monetary equilibria}. However, in this case, the set in (\ref{AlternativeDefParetoOptimal}) would be a subset of $\mathcal{H}^{PO}$. Furthermore, if this sort of assumption is fully dropped, then the results in Section \ref{sec3} still furnish the existence of \textit{monetary equilibria}.} (e.g., uniform bounds on the Gaussian curvature of the indifference curves \parencite[p. 296]{BalaskoShell_1980} and bounded coefficients of smoothness and strictness \parencite[p. 803]{OkunoZilcha_1980}).

Since $\mathcal{H}$ is seen as the set of equilibria that represents long-standing economic activity, notice that, under Assumption \ref{assCassCriterion}, $\mathcal{H}^{PO}$ can be regarded as the set of equilibria that represents long-standing, efficient economic activity, with $\mathcal{E}$ being just part of a larger classical economy whose inception was in period $t=T$, $T\leq -1$.

This is because, in this larger classical economy, competitive markets under perfect foresight yield a Pareto optimal equilibrium through prices $(p_{-T},\ldots,p_{-1},p)\in \mathbb{R}^{\infty}_{++}$, and $ \sum^{\infty}_{t=-T} (H_{t}\Vert p_{t}\Vert)^{-1}=+\infty$ because of the \textcite{Cass_1972} criterion. Also, due to the equivalence
\begin{eqnarray*}
    \sum^{\infty}_{t=-T} \frac{1}{H_{t}\Vert p_{t}\Vert}=+\infty\iff \sum^{\infty}_{t=0} \frac{1}{H_{t}\Vert p_{t}\Vert}=+\infty,
\end{eqnarray*}
every efficient equilibrium of the larger economy provides an element of $\mathcal{H}^{PO}$. Since no specific form of the larger economy is given, $\mathcal{H}^{PO}$ contains the projection of the sets of efficient equilibria of all possible larger classical economies that satisfy Assumptions \ref{assBoundsPopAndEndownments}--\ref{assCassCriterion}. Notice also that every element of $\mathcal{H}^{PO}$ lies within an efficient equilibrium of some larger classical economy that satisfies Assumptions \ref{assBoundsPopAndEndownments}--\ref{assCassCriterion}.

Before moving on to the results of this section, the following definition of short-sighted equilibrium for the economy $\mathcal{E}$ is needed.
\begin{definition}\label{defSetJSightedEquilibria}
The \textit{set of $j$-sighted equilibria} $\mathcal{V}_{j}\subset\mathbb{R}^{\sum^{j+1}_{i=0}L_{i}}_{++}$, $j\geq1$, of economy $\mathcal{E}$ is
\begin{eqnarray*}
\mathcal{V}_{j}=\{p\in\mathbb{R}^{\sum^{j+1}_{i=0}L_{i}}_{++}\mid p\in\mathcal{Z}_{j}, (p_{0},p_{1})\in\mathcal{B}_{0}(\sigma_{0}),(p_{j}/p_{j1},p_{j+1}/p_{j1})\in\mathcal{B}_{j}(\sigma_{j})\}.
\end{eqnarray*}
\end{definition}

Definition \ref{defSetJSightedEquilibria} focuses only on ``short-term'' market clearing, limited to a finite number of periods. Furthermore, prices $(p_{0},p_{1})$ and $(p_{j},p_{j+1})$ that the generations on the edges face must be compatible with past and future overlaps according to Assumptions \ref{assBoundsPopAndEndownments} and \ref{assBoundsPrices}. For the next lemma, let $\mathcal{V}^{\infty}_{j}=\mathcal{V}_{j}\times\mathbb{R}^{\infty}_{+}$, $j\geq1$, be the natural embedding of $\mathcal{V}_{j}$ into $\mathbb{R}^{\infty}_{+}$.

\begin{lemma}\label{lemmaVjInfinite}
Under Assumptions \ref{assBoundsPopAndEndownments}--\ref{assBoundsPrices}, the set of j-sighted equilibria $\mathcal{V}_{j}\subset\mathbb{R}^{\sum^{j+1}_{i=0}L_{i}}_{++}$, $j\geq1$, is non-empty and compact, $\mathcal{V}_{j}\subseteq\prod^{j+1}_{t=0}\mathcal{K}_{t}$ and $\mathcal{V}^{\infty}_{j+1}\subseteq \mathcal{V}^{\infty}_{j}$.
\end{lemma}

Lemma \ref{lemmaVjInfinite} states that $\{\mathcal{V}^{\infty}_{j}\}_{j\geq1}$ is a nested sequence of non-empty sets. Let $\lim_{j\rightarrow\infty}\mathcal{V}^{\infty}_{j}=\bigcap_{j\geq1}\mathcal{V}^{\infty}_{j}$. The next result clarifies the relation between the sequence $\{\mathcal{V}^{\infty}_{j}\}_{j\geq1}$ and the set of equilibria $\mathcal{H}$.

\begin{theorem}\label{theoEqSetCompact}
Under Assumptions \ref{assBoundsPopAndEndownments}--\ref{assBoundsPrices}, the set of equilibria $\mathcal{H}\subset\mathbb{R}^{\infty}_{++}$ is non-empty and compact, with $\mathcal{H}=\lim_{j\rightarrow\infty}\mathcal{V}^{\infty}_{j}\subseteq\mathcal{K}$. 
\end{theorem}

In order to better characterize the subset of Pareto optimal equilibria $\mathcal{H}^{PO}\subseteq\mathcal{H}$, yet another definition is needed.

\begin{definition}\label{defRealSavings}
The \textit{real savings} $s^{h}:\mathbb{R}^{L_{t}+L_{t+1}}_{++}\rightarrow \mathbb{R}$ of household $h\in G_{t}$, $t\geq0$, is
    \begin{eqnarray*}
       s^{h}(p_{t},p_{t+1}) = \frac{p_{t}}{\Vert p_{t}\Vert}\cdot (e^{h}_{t}-x^{h}_{t}(p_{t},p_{t+1})).
    \end{eqnarray*}
\end{definition}

Definition \ref{defRealSavings} can be understood through the households budget constraint. Since the utility function of $h\in G_{t}$, $t\geq0$, is locally nonsatiated, the budget constraint holds with equality (i.e., Walras' law) and all resources not consumed when young $p_{t}\cdot (e^{h}_{t}-x^{h}_{t}(p_{t},p_{t+1}))>0$ are destined for old-age consumption $p_{t+1}\cdot (x^{h}_{t+1}(p_{t},p_{t+1})-e^{h}_{t+1})>0$. This description assumes a positive savings level, but we could have zero or negative savings depending on the price levels. 

The next remark is that, to obtain a measure of real savings, a reference basket of commodities at time $t$ must be defined since $p_{t}\cdot (e^{h}_{t}-x^{h}_{t}(p_{t},p_{t+1}))$ is homogeneous of degree one and varies with the price level (i.e., it is a nominal value). A possibility is the basket containing one unit of all existing commodities $(1,\ldots,1)\in\mathbb{R}^{L_{t}}_{++}$, which costs $\sum^{L_{t}}_{i=1}p_{ti}=\Vert p_{t}\Vert$. Dividing $p_{t}\cdot (e^{h}_{t}-x^{h}_{t}(p_{t},p_{t+1}))$ by $\Vert p_{t}\Vert$ thus yields the desired measure of real savings. The next result gives a first relation between real savings and the subset of Pareto optimal equilibria $\mathcal{H}^{PO}$.

\begin{lemma}\label{lemmaFirstParetoOptEquilibria}
Under Assumptions \ref{assBoundsPopAndEndownments}, \ref{assBoundsPrices} and \ref{assCassCriterion}, if $p\in\mathcal{H}/\mathcal{H}^{PO}$, then 
\begin{eqnarray*}
    \lim_{t\rightarrow\infty}\sum_{h\in G_{t}}\frac{s^{h}(p_{t},p_{t+1})}{H_{t}}=0.
\end{eqnarray*}
\end{lemma}

Lemma \ref{lemmaFirstParetoOptEquilibria} states that an inefficient equilibrium depresses real savings per capita, leading to an economic scenario in which real interperiod wealth transfers become smaller for subsequent generations, and intraperiod trade or autarky are all that remain in the long run. In order to find conditions under which this depressed real savings dynamic fully characterizes inefficient equilibria, I restrict the analysis to economies in which positive real savings per capita can be considered a natural outcome.

\begin{definition}\label{defProneSavings}
The economy $\mathcal{E}$ is \textit{prone to savings} if there are $\varepsilon>0$, $\delta>0$ such that, for $t\geq0$, $(p_{t},p_{t+1})\in\mathbb{R}^{L_{t}+L_{t+1}}_{++}$ and $(p_{t}/p_{t1},p_{t+1}/p_{t1})\in\mathcal{B}_{t}(\sigma_{t})$,
\begin{eqnarray*}
    \sum_{h\in G_{t}} \frac{s^{h}(p_{t},p_{t+1})}{H_{t}}\leq \delta\implies \frac{\Vert p_{t} \Vert}{\Vert p_{t+1} \Vert} \leq \frac{\alpha_{t}}{1+\varepsilon}.
\end{eqnarray*}
\end{definition}

Definition \ref{defProneSavings} states that the economy $\mathcal{E}$ is prone to savings if, for prices compatible with the bounds on endowments and demographic dynamic, low real savings per capita from generation $G_{t}$, $t\geq0$, imply a low real return rate relative to the population growth rate. Conversely, if the real return rate is sufficiently high compared to the demographic growth rate, then the real savings per capita will also be high. 

This is because the utility maximization problem of household $h\in G_{t}$, $t\geq0$, can be written as
\begin{eqnarray*}
    \max_{(c,s)\in\mathbb{R}^{L_{t}+L_{t+1}}_{+}\times \mathbb{R}} & u^{h}(c)\\
    \textit{s.t.} & \frac{p_{t}}{\Vert p_{t}\Vert}\cdot (c_{t}-e^{h}_{t})+ s=0\\
    & \frac{p_{t+1}}{\Vert p_{t+1}\Vert}\cdot (c_{t+1}-e^{h}_{t+1})\leq\frac{\Vert p_{t}\Vert}{\Vert p_{t+1}\Vert} s= r_{t}s,
\end{eqnarray*}
with $r_{t}=\Vert p_{t}\Vert/\Vert p_{t+1}\Vert$ representing the real return rate faced by generation $G_{t}$. 

We can infer from Definition \ref{defProneSavings} that high demographic growth rates favor the emergence of prone-to-savings economies. We can also guess that households' preferences that place a high value on old-age consumption and skewed endowment distribution toward the younger periods of life are both situations that favor the emergence of prone-to-savings economies, since in these cases larger real return rates tend to increase real savings per capita.

Therefore, a prone-to-savings economy can be thought of as one in which generations have, given the demographic dynamic, a meaningful preference for old age consumption or a sufficiently skewed endowment distribution towards the earlier periods of life. In these economies, there is a natural propensity to save, and the real savings level becomes a signal of whether the terms of trade between younger and older periods of life are unfavorable, since unnaturally low savings could only be achieved through unfavorable real return rates. 

The next result formalizes this intuitive description by providing a sufficient condition for the economy $\mathcal{E}$ to be prone to savings when utility functions display constant elasticity of substitution (CES) and are homogeneous within each generation.

\begin{proposition}\label{propSufficientConditionProneSavings}
Under Assumptions \ref{assBoundsPopAndEndownments} and \ref{assBoundsPrices}, suppose there are $(\lambda_{t},\tilde{\lambda}_{t+1})\in\mathbb{R}^{L_{t}+L_{t+1}}_{+}$, $t\geq0$, with $\Vert \lambda_{t}\Vert+\Vert \tilde{\lambda}_{t+1}\Vert=1$, and $\rho_{t}\in(0,1)$ such that, for all $h\in G_{t}$, 
\begin{eqnarray*}
    u^{h}(c_{t},c_{t+1})=\biggr(\sum^{L_{t}}_{i=1}\lambda_{ti}c^{\, \rho_{t}}_{ti}+\sum^{L_{t+1}}_{i=1}\tilde{\lambda}_{(t+1)i}c^{\, \rho_{t}}_{(t+1)i}\biggr)^{1/\rho_{t}}.
\end{eqnarray*}
Also, let $\gamma_{t}\in[0,1]$, $t\geq0$, be defined as
\begin{eqnarray*}
\gamma_{t}=\frac{\sum^{L_{t}}_{i=1}\lambda^{\eta_{t}}_{ti}}{\sum^{L_{t}}_{i=1}\lambda^{\eta_{t}}_{ti}+ \sigma^{2(\eta_{t}-1)}_{t}\sum^{L_{t+1}}_{i=1}\tilde{\lambda}^{\eta_{t}}_{(t+1)i}},   
\end{eqnarray*}
with $\eta_{t}=1/(1-\rho_{t})>1$. If 
\begin{eqnarray}\label{eqPropProneSavingsSuff}
\inf_{t\geq0}\{(1-\gamma_{t})\min_{1\leq i\leq L_{t}}e^{t}_{ti}-\gamma_{t}\Vert e^{t}_{t+1}\Vert_{\infty}\alpha^{-1}_{t}\}>0,
\end{eqnarray}
then the economy $\mathcal{E}$ is prone to savings.
\end{proposition}

Proposition \ref{propSufficientConditionProneSavings} can be interpreted from three different angles based on (\ref{eqPropProneSavingsSuff}). First, sufficiently skewed aggregate endowment distributions\footnote{The relation between skewed endowment distributions and efficiency of equilibria has also been noticed by \textcite{DucGhiglino_1998}.}, represented by the difference between the terms containing $\min_{1\leq i\leq L_{t}}e^{t}_{ti}$ and $\Vert e^{t}_{t+1}\Vert_{\infty}$, lead towards a prone-to-savings economy. In addition, sufficiently high demographic growth rates $\{\alpha_{t}\}_{t\geq0}$ and meaningfully valued consumption in older periods, represented by small values of $\{\gamma_{t}\}_{t\geq0}$, also lead towards a prone-to-savings economy.

I say ``lead towards'' because this comparative statics ignores the possible effects of changing preferences, endowments, and demographic dynamics on the values of $\sigma_{t}$, $t\geq0$, given by Assumption \ref{assBoundsPrices}. If, however, the value of $\gamma_{t}$ does not depend on $\sigma_{t}$, the comparative statics exercise can be made effortlessly. This happens when $\rho_{t}\rightarrow 0$, $t\geq0$, and therefore all utility functions are log-linear\footnote{The case of overlapping generations economies with log-linear utility functions was extensively studied by \textcite{BalaskoShell_1981_2}.}. Due to its tractability, this case is stated as a corollary of Proposition \ref{propSufficientConditionProneSavings}. 

\begin{corollary}\label{corSuffConditionProneSavings}
Under Assumptions \ref{assBoundsPopAndEndownments} and \ref{assBoundsPrices}, suppose there are $(\lambda_{t},\tilde{\lambda}_{t+1})\in\mathbb{R}^{L_{t}+L_{t+1}}_{+}$, $t\geq0$, with $\Vert \lambda_{t}\Vert+\Vert \tilde{\lambda}_{t+1}\Vert=1$, such that, for all $h\in G_{t}$, $u^{h}(c_{t},c_{t+1})=\sum^{L_{t}}_{i=1}\lambda_{ti}\ln c_{ti}+\sum^{L_{t+1}}_{i=1}\tilde{\lambda}_{(t+1)i}\ln c_{(t+1)i}$. If $\inf_{t\geq0}\{\Vert\tilde{\lambda}_{t+1}\Vert\min_{1\leq i\leq L_{t}}e^{t}_{ti}-\Vert \lambda_{t}\Vert\Vert e^{t}_{t+1}\Vert_{\infty}\alpha^{-1}_{t}\}>0$, then the economy $\mathcal{E}$ is prone to savings.   
\end{corollary}

I now state the final assumption made in this paper.  

\begin{assumption}\label{assProneSavingsEconomy}
The economy $\mathcal{E}$ is prone to savings.
\end{assumption}

The next result provides an essential relation between real savings per capita and the subset of Pareto optimal equilibria $\mathcal{H}^{PO}\subseteq\mathcal{H}$ in prone-to-savings economies.

\begin{lemma}\label{lemmaSecondParetoOptEquilibria}
Under Assumptions \ref{assBoundsPopAndEndownments} and \ref{assBoundsPrices}--\ref{assProneSavingsEconomy}, if $p\in\mathcal{H}$ and $\sum_{h\in G_{t}}s^{h}(p_{t},p_{t+1})/H_{t}\leq \delta$, for some $t\geq0$,\footnote{To be fully clear, $\delta>0$ is given by Assumption \ref{assProneSavingsEconomy} and Definition \ref{defProneSavings}.} then $p\notin\mathcal{H}^{PO}$.
\end{lemma}

Lemma \ref{lemmaSecondParetoOptEquilibria} states that, for a prone-to-savings economy $\mathcal{E}$, if, for any generation, the real savings per capita falls below the threshold value $\delta>0$, then the respective equilibrium cannot be Pareto optimal. Due to the propensity to savings of $\mathcal{E}$, Lemma \ref{lemmaSecondParetoOptEquilibria} simply reveals that a Pareto optimal equilibrium must be aligned with this nature in all periods. 

The next result gathers Lemma \ref{lemmaFirstParetoOptEquilibria} and Lemma \ref{lemmaSecondParetoOptEquilibria} to provide, through the dynamic of real savings per capita, a complete characterization of the subset of Pareto optimal equilibria in prone-to-savings economies.
\begin{proposition}\label{propNecAndSuffConditionParetoOpt}
Let $p\in\mathcal{H}$. Under Assumptions \ref{assBoundsPopAndEndownments} and \ref{assBoundsPrices}--\ref{assProneSavingsEconomy}, $p\notin\mathcal{H}^{PO}$ if, and only if, 
\begin{eqnarray*}
    \lim_{t\rightarrow\infty}\sum_{h\in G_{t}}\frac{s^{h}(p_{t},p_{t+1})}{H_{t}}=0.
\end{eqnarray*}
\end{proposition}

Proposition \ref{propNecAndSuffConditionParetoOpt} states that nonvanishing real savings per capita completely characterize efficient equilibria in prone-to-savings economies. It reveals that the Cass criterion, which is seen as a statement on the dynamics of prices (e.g., that prices do not grow too fast), is, in prone-to-savings economies, equivalent to nonvanishing real savings per capita (which can be seen as a sort of transversality condition).

Furthermore, by inspecting the proof of Lemma \ref{lemmaSecondParetoOptEquilibria} at the \hyperref[appn]{Appendix}, one can see that by defining $\mathcal{S}_{j}\subseteq\mathcal{H}$, $j\geq0$, as
\begin{eqnarray}\label{eqSj}
    \mathcal{S}_{j}=\biggr\{p\in\mathcal{H} \mid \sum_{h\in G_{j}}\frac{s^{h}(p_{j},p_{j+1})}{H_{j}}\leq \delta\biggr\},
\end{eqnarray}
we have $\mathcal{S}_{j}\subseteq\mathcal{S}_{j+1}$, $j\geq0$. Theorem \ref{theoEqSetCompact} and the continuity of real savings per capita of generation $G_{j}$ imply that $\mathcal{S}_{j}$ is compact, $j\geq0$. Therefore, $\{\mathcal{S}_{j}\}_{j\geq1}$ is an increasing sequence of compact sets, and Proposition \ref{propNecAndSuffConditionParetoOpt} implies that $\bigcup_{j\geq0}\mathcal{S}_{j}=\mathcal{H}/\mathcal{H}^{PO}$. 

I proceed with the following existence result.

\begin{theorem}\label{theoExistenceOfParetoOptJME}
Under Assumptions \ref{assBoundsPopAndEndownments}--\ref{assProneSavingsEconomy}, the subset of Pareto optimal equilibria $\mathcal{H}^{PO}\subseteq\mathcal{H}$ is non-empty and compact, with 
\begin{eqnarray}\label{eqMonetary}
    \sum_{h\in G_{0}}\frac{s^{h}(p_{0},p_{1})}{H_{0}}>\delta,
\end{eqnarray}
for $p\in\mathcal{H}^{PO}$.
\end{theorem}

Theorem \ref{theoExistenceOfParetoOptJME} provides sufficient conditions for the existence of a Pareto optimal equilibrium in overlapping generations economies and, under such conditions, the subset of Pareto optimal equilibria $\mathcal{H}^{PO}$ is compact. Also, (\ref{eqMonetary}) implies that all these optimal equilibria are \textit{monetary} due to the following reasoning.

Let $p\in\mathcal{H}^{PO}$ and $\mathcal{E}^{3}$ be the economy\footnote{See, in the beginning of this Section, the discussion on the three types of overlapping generations economies $\mathcal{E}^{1}$, $\mathcal{E}^{2}$, and $\mathcal{E}^{3}$ with one-period-lived households in $t=1$, the ``inception of the economy.''} formed by all generations $G_{t}$, $t\geq1$, and a generation $G^{3}_{0}$ of one-period-lived households. For each $h\in G_{0}$, there is a household $h^{\prime}\in {G}^{3}_{0}$ with a utility function $u^{h^{\prime}}(c_{1})=u^{h}(x^{h}_{0}(p_{0},p_{1}),c_{1})$, an endowment $e^{h^{\prime}}_{1}=e^{h}_{1}\in\mathbb{R}^{L_{1}}_{+}$ and inside or outside money holdings or debt given by $m^{h^{\prime}}=\Vert p_{0}\Vert s^{h}(p_{0},p_{1})\in\mathbb{R}$. Notice that $p\in\mathcal{H}^{PO}$ and
\begin{eqnarray*}
    x^{h}_{1}(p_{0},p_{1})=x^{h^{\prime}}_{1}(1,p_{1})=\textrm{argmax}_{c_{1}\in\mathbb{R}^{L_{1}}_{+}}& u^{h^{\prime}} (c_{1})\\
    \textrm{ s.t. }& p_{1}\cdot (c_{1}-e^{h^{\prime}}_{1})\leq m^{h^{\prime}},
\end{eqnarray*}
imply that $(p_{m},p_{1},\ldots)=(1,p_{1},\ldots)$ is an equilibrium of $\mathcal{E}^{3}$. Let $m^{h^{\prime}}_{+}=\max\,\{0,m^{h^{\prime}}\}$, $h^{\prime}\in G^{3}_{0}$. If $m^{h^{\prime}}\geq0$, $h^{\prime}\in G^{3}_{0}$, then this is a classical monetary equilibrium, since $p_{m}=1$ and the total stock of fiat money is strictly positive due to (\ref{eqMonetary}), since
\begin{eqnarray*}
    0<\Vert p_{0}\Vert H_{0} \delta<\sum_{h\in{G}_{0}}\Vert p_{0}\Vert s^{h}(p_{0},p_{1})=\sum_{h^{\prime}\in G^{3}_{0}}m^{h^{\prime}}=\sum_{h^{\prime}\in G^{3}_{0}}m^{h^{\prime}}_{+}.
\end{eqnarray*}
Notice that the last equality states that every increase in income from households in generation $G^{3}_{0}$ at $t=1$ is due to fiat money holdings.

If $m^{h^{\prime}}<0$, for some $h^{\prime}\in G^{3}_{0}$, then we have
\begin{eqnarray*}
    0<\sum_{h^{\prime}\in G^{3}_{0}}m^{h^{\prime}}<\sum_{h^{\prime}\in G^{3}_{0}}m^{h^{\prime}}_{+}.
\end{eqnarray*}
In this case, we also have a monetary equilibrium (i.e., $p_{m}=1$) with a strictly positive stock of fiat money due to the first inequality. 

The second inequality, however, states that increases in income from households in generation $G^{3}_{0}$ are due either to fiat money holdings (which sum up to $\sum_{h^{\prime}\in G^{3}_{0}}m^{h^{\prime}}>0$) or to inside money holdings (which sum up to $\sum_{h^{\prime}\in G^{3}_{0}}m^{h^{\prime}}_{+}-m^{h^{\prime}}>0$). Therefore, unlike the classical monetary equilibrium setup, households from generation $G^{3}_{0}$ have an amount of IOU debt outstanding at $t=1$ that sums up to $\sum_{h^{\prime}\in G^{3}_{0}}m^{h^{\prime}}_{+}-m^{h^{\prime}}>0$.  

It remains to demonstrate how, exactly, Theorem \ref{theoExistenceOfParetoOptJME} sets aside the nonexistence examples given by \textcite[Section V, pp. 57-70]{CassOkunoZilcha_1979}\footnote{See, for another nonexistence example, \textcite[pp. 217-219]{OkunoZilcha_1982}.}, since this result is not directly applicable to a classical economy fiat with money $\mathcal{E}^{2}$.

The first thing to notice is that the nonexistence examples in \textcite[Section V, pp. 57-70]{CassOkunoZilcha_1979} rely on one of four assumptions: ``completely inflexible tastes'' (i.e., Leontief utility function) or ``completely skewed endowments'' towards the older period of life  \parencite[Section V-A, p. 58]{CassOkunoZilcha_1979}; local satiation \parencite[Section V-B, p. 61]{CassOkunoZilcha_1979}; or non-convex preferences \parencite[Section V-C, p. 70]{CassOkunoZilcha_1979}. 

Therefore, the assumptions we have made about the overlapping generations economy $\mathcal{E}$ already set us apart from these four assumptions. In particular, notice that a prone-to-savings economy cannot have households with completely skewed endowments towards the older periods of life.

However, we still need to show that Theorem \ref{theoExistenceOfParetoOptJME} ensures the existence of an optimal monetary equilibrium for a classical economy with fiat money. Let $\mathcal{E}^{2}$ be such an economy, with household $h\in G^{2}_{0}$ characterized by a continuous, non-decreasing, semi-strictly quasi-concave utility function $u^{h}:\mathbb{R}^{L_{1}}_{+}\rightarrow\mathbb{R}$ without local maxima, a nonzero endowment $e^{h}_{1}\in\mathbb{R}^{L_{1}}_{+}$ and an amount of fiat money $m^{h}\geq0$, with $\sum_{h\in G^{2}_{0}}m^{h}=1$. Also, the aggregate endowment satisfies $\sum_{t\geq0}\sum_{h\in G^{2}_{t}}e^{h}\in\mathbb{R}^{\infty}_{++}$. Furthermore, let $\mathcal{E}^{2}_{\geq1}$ be the economy formed by all generations of $\mathcal{E}^{2}$ born in $t\geq1$. Finally, I recall that
\begin{eqnarray*}
    x^{h}_{1}(p_{m},p_{1})=\textrm{argmax}_{c\in\mathbb{R}^{L_{1}}_{+}}& u^{h} (c) \\
    \textrm{ s.t. }& p_{1}\cdot (c-e^{h}_{1})\leq p_{m}m^{h},
\end{eqnarray*}
for $(p_{m},p_{1})\in\mathbb{R}_{+}\times\mathbb{R}^{L_{1}}_{++}$, $h\in G^{2}_{0}$. 

\begin{theorem}\label{theoClassModelExistence}
Let $\mathcal{E}^{2}$ be a classical economy with fiat money. Suppose $\mathcal{E}^{2}$ satisfies Assumptions \ref{assBoundsPopAndEndownments}, \ref{assResourceRelated} and \ref{assCassCriterion}; $\mathcal{E}^{2}_{\geq1}$ satisfies Assumptions $\ref{assBoundsPrices}$ and \ref{assProneSavingsEconomy}; $\alpha_{\min}\leq H_{1}/(H_{0}+1)\leq\alpha_{\max}$; $m^{h}\leq e_{\max}$, $h\in G^{2}_{0}$; there is $h\in G_{0}^{2}$ with $\Vert e^{h}_{1}\Vert_{\infty}<e_{\max}$; and there is $\lambda>0$ such that\footnote{It is worth highlighting that (\ref{eqLowerBoundG20}) requires that, for prices $p_{1}\in\mathbb{R}^{L_{1}}_{++}$ stated in terms of the numeraire (i.e., $p_{m}=1$), if some commodity is too cheap, then the aggregate demand of the older generation becomes incompatible with Assumption \ref{assBoundsPopAndEndownments}.}
\begin{eqnarray}\label{eqLowerBoundG20}
    \min_{1\leq i\leq L_{1}}\, p_{1i}<\lambda\implies \biggr\Vert \sum_{h\in G^{2}_{0}}\frac{x^{h}_{1}(1,p_{1})}{H_{0}}\biggr\Vert_{\infty}>\beta e_{\max},
\end{eqnarray}
for $p_{1}\in\mathbb{R}^{L_{1}}_{++}$. Then, there exists a Pareto optimal monetary equilibrium of $\mathcal{E}^{2}$.
\end{theorem}

Theorem \ref{theoClassModelExistence} helps to fill a long-standing gap in the literature by providing sufficient conditions under which a classical non-stationary overlapping generations economy with fiat money possesses a Pareto optimal monetary equilibrium. In short, it states that if the economy is prone to savings, then money has value and leads the economy to an efficient allocation, thus furnishing a solution to the \textcite{Hahn_1965} problem.

Although Theorems \ref{theoExistenceOfParetoOptJME} and \ref{theoClassModelExistence} imply the existence of Pareto optimal monetary equilibria, they do not provide a direct way to calculate or, in some sense, approximate them. Since this is a fundamental milestone for properly deploying these overlapping generations models, the next section describes a backward calculation algorithm to find elements from $\mathcal{H}^{PO}$.

\section{The backward calculation algorithm}\label{sec4}
This section defines a backward calculation algorithm conceived to identify elements from the subset of Pareto optimal equilibria $\mathcal{H}^{PO}\subseteq\mathcal{H}$ of a given prone-to-savings economy $\mathcal{E}$. I start with the following definition.

\begin{definition}\label{defFiniteRep}
Let $\mathcal{E}$ and $\{\mathcal{E}_{k}\}_{k\geq1}$ satisfy Assumptions \ref{assBoundsPopAndEndownments}--\ref{assProneSavingsEconomy}. Then, $\{\mathcal{E}_{k}\}_{k\geq1}$ is \textit{finitely replicating} $\mathcal{E}$ if: (i) Assumption \ref{assBoundsPopAndEndownments} is satisfied for common values of $\alpha_{\min},\alpha_{\max}$ and $e_{\max}>0$; (ii) there exists a compact set $\mathcal{K}^{\prime}\subset\mathbb{R}^{\infty}$ such that $\mathcal{H}\subseteq\mathcal{K}^{\prime}$ and $\mathcal{H}_{k}\subseteq\mathcal{K}^{\prime}$, $k\geq1$; and (iii) all generations of $\mathcal{E}_{k}$ born until period $k\geq1$ are identical to those of $\mathcal{E}$. In this case, the economy $\mathcal{T}_{k}$ formed by all generations of $\mathcal{E}_{k}$ born after period $k\geq1$ is called a \textit{tail economy}.
\end{definition}

The notion of \textit{finite replication} is rather intuitive since it only requires the replicating economies $\{\mathcal{E}_{k}\}_{k\geq1}$ to coincide with the reference economy $\mathcal{E}$ over finite but arbitrarily large time horizons and that Assumptions \ref{assBoundsPopAndEndownments}--\ref{assProneSavingsEconomy} are met with some uniformity\footnote{Notice that Theorem \ref{theoEqSetCompact} implies that $\mathcal{H}$ and $\mathcal{H}_{k}$, $k\geq1$, are already contained in compact sets. Therefore, condition (ii) from Definition \ref{defFiniteRep} requires a \textit{unique} compact set that contains all of them. After some careful thought, one can notice that this is actually a constraint over the tail economies.}. Definition \ref{defFiniteRep} also implies that the sequence of tail economies $\{\mathcal{T}_{k}\}_{k\geq1}$ completely defines the finitely replicating sequence $\{\mathcal{E}_{k}\}_{k\geq1}$. 

As previously mentioned, Theorems \ref{theoExistenceOfParetoOptJME} and \ref{theoClassModelExistence} do not provide a way to actually find a Pareto optimal monetary equilibrium of $\mathcal{E}$. One possible path in this direction would be to calculate all elements of $\mathcal{H}$ and then select those that satisfy the \textcite{Cass_1972} criterion. However, in general, exactly solving an infinite number of equilibrium equations and then evaluating the convergence of an infinite series is impossible. Traditionally, one way to ``cut the Gordian knot'' \parencite[p. 470]{Samuelson_1958} is to work with stationary economies, looking for steady states and locally convergent solutions after linearizing the equilibrium equations around them (e.g., \textcite[pp. 222-223]{KehoeLevine_1990}).

A closer look at Assumption \ref{assCassCriterion} indicates an alternative path that can be applied to a non-stationary economy, a path that is related to truncations of the model (e.g., \textcite[p. 223]{KehoeLevine_1990}). First, observe that if $\sum^{+\infty}_{t=0}(H_{t}\Vert p_{t}\Vert)^{-1}=+\infty$, then $\sum^{+\infty}_{t=k}(H_{t}\Vert p_{t}\Vert)^{-1}=+\infty$, for all $k\geq1$. Therefore, to check if $p\in\mathcal{H}^{PO}$, we could simply ignore an arbitrarily large number of periods.

However, if it is easier to solve the equilibrium equations after a given period $k\geq1$, we could simply start by solving these equations in the tail of the economy to find $(p_{k},p_{k+1},\ldots)\in\mathbb{R}^{\infty}_{++}$ and then select the ones for which $\sum^{+\infty}_{t=k}(H_{t}\Vert p_{t}\Vert)^{-1}=+\infty$. With this set of selected tail prices, we could use the equilibrium equations to move backward and calculate all possible previous prices $(p_{0},\ldots,p_{k-1})$. Since Theorem \ref{theoExistenceOfParetoOptJME} implies $\mathcal{H}^{PO}\neq \emptyset$, the algorithm will reach its end for at least one selected price.

Clearly, depending on the overlapping generations economy, solving the equilibrium equations after period $k\geq1$ may be no less difficult than solving all equilibrium equations since, in any case, we are dealing with infinities. Nevertheless, it seems reasonable to guess that, under natural circumstances, what will happen in a few years from now is more relevant for today's economy than what will happen in a few decades, and what will happen in a few decades is more relevant than what will happen in a century from now. 

This reasoning leads us to guess that replacing generations that live in the far future may have a small effect on today's economic variables we are interested in and can make it easier to solve the infinite equilibrium equations remaining after the start of the ``far future.''

In terms of Definition \ref{defFiniteRep}, replacing future generations is equivalent to appropriately choosing tail economies $\{\mathcal{T}_{k}\}_{k\geq1}$ and building a finitely replicating sequence $\{\mathcal{E}_{k}\}_{k\geq1}$ for $\mathcal{E}$. These tail economies must be well behaved in the sense that we know exactly what their sets of Pareto optimal monetary equilibria $\{\mathcal{H}^{PO}_{\tau k}\}_{k\geq1}$ (the symbol $\tau$ stands for ``tail'') are, so that we do not need to calculate them. Actually, we intend to depart from these well-known sets of Pareto optimal monetary equilibria and use equilibrium equations to move backward.

In this setting, due to Theorem \ref{theoExistenceOfParetoOptJME}, $\mathcal{H}^{PO}_{k}\neq\emptyset$ and our backward calculation algorithm will reach its end after departing from at least one equilibrium of $\mathcal{H}^{PO}_{\tau k}$, $k\geq1$. However, even with the algorithm effectively finding $p_{k}\in\mathcal{H}^{PO}_{k}$, there is no guaranty that this result will be related to the set of equilibria $\mathcal{H}$ of the reference economy $\mathcal{E}$ or to its subset of efficient monetary equilibria $\mathcal{H}^{PO}\subseteq\mathcal{H}$. Our aim is to address these concerns and demonstrate the sense in which this algorithm is valid for finding elements of $\mathcal{H}^{PO}$. The next result is the first step in this direction.

\begin{lemma}\label{lemmaTailEquilibria}
Let the sequence $\{\mathcal{E}_{k}\}_{k\geq1}$ be finitely replicating $\mathcal{E}$, with $\{\mathcal{T}_{k}\}_{k\geq1}$ the sequence of tail economies and $\mathcal{H}_{\tau k}$ their respective sets of equilibria. Then, for all $k\geq1$, $(p_{0},p_{1},\ldots)\in\mathcal{H}_{k}$ implies $(p_{k+1}/p_{(k+1)1},p_{k+2}/p_{(k+1)1},\ldots)\in\mathcal{H}_{\tau k}$, and $(p_{0},p_{1},\ldots)\in\mathcal{H}^{PO}_{k}$ implies $(p_{k+1}/p_{(k+1)1},p_{k+2}/p_{(k+1)1},\ldots)\in\mathcal{H}^{PO}_{\tau k}$.
\end{lemma}

Lemma \ref{lemmaTailEquilibria} formalizes the idea that, when looking for the set of equilibria $\mathcal{H}_{k}$ of the economy $\mathcal{E}_{k}$, one can start from the well-known set of equilibria $\mathcal{H}_{\tau k}$ of the tail economy $\mathcal{T}_{k}$, $k\geq1$, and, by backward calculation through equilibrium equations from time $k+1$ to $1$, find all elements from $\mathcal{H}_{k}$ up to normalization. Since Theorem \ref{theoEqSetCompact} implies $\mathcal{H}_{k}\neq\emptyset$, this algorithm leads to at least one equilibrium of $\mathcal{H}_{k}$.

Also, due to Theorem \ref{theoExistenceOfParetoOptJME}, this same algorithm remains valid when we deal with the subset of Pareto optimal monetary equilibria $\mathcal{H}^{PO}_{k}\subseteq \mathcal{H}_{k}$ of the economy $\mathcal{E}_{k}$ and the well-known subset of Pareto optimal monetary equilibria $\mathcal{H}^{PO}_{\tau k}\subseteq \mathcal{H}_{\tau k}$ of the tail economy $\mathcal{T}_{k}$, $k\geq1$.

The next lemma relates the sets of equilibria $\mathcal{H}_{k}$ of the economy $\mathcal{E}_{k}$, $k\geq1$, obtained through the backward calculation algorithm, and the set of equilibria $\mathcal{H}$ of the reference economy $\mathcal{E}$.

\begin{lemma}\label{lemmaFiniteReplEqSet}
If $\{\mathcal{E}_{k}\}_{k\geq1}$ is finitely replicating $\mathcal{E}$, then 
\begin{eqnarray*}
    \lim_{j\rightarrow\infty}\overline{\bigcup_{k\geq j}\mathcal{H}_{k}}\subseteq \mathcal{H},
\end{eqnarray*}
with $\lim_{j\rightarrow\infty}\overline{\bigcup_{k\geq j}\mathcal{H}_{k}}$ a non-empty compact set.
\end{lemma}

Lemma \ref{lemmaFiniteReplEqSet} is a result\footnote{By inspecting the proof of Lemma \ref{lemmaFiniteReplEqSet}, one may notice that Assumption \ref{assProneSavingsEconomy}, although mentioned in Definition \ref{defFiniteRep}, is not necessary for this result (i.e., the result is valid even if the overlapping economies are not prone to savings).} regarding the approximation of equilibria from $\mathcal{H}$ by equilibria from $\mathcal{H}_{k}$, $k\geq1$. It agrees with the idea that if $\{\mathcal{E}_{k}\}_{k\geq1}$ is finitely replicating $\mathcal{E}$, then the tail economies become less relevant each time, and solving an arbitrarily large number of identical equilibrium equations leads approximately to an equilibrium of $\mathcal{E}$. 

Also, the fact that $\lim_{j\rightarrow\infty}\overline{\bigcup_{k\geq j}\mathcal{H}_{k}}\neq\emptyset$ ensures that one will indeed find at least one equilibrium of $\mathcal{H}$ as the limit of a sequence $\{p^{k}\}_{k\geq1}$, with $p^{k}\in\mathcal{H}_{k}$, $k\geq1$.

The question that arises naturally is whether this limit procedure can also be applied to find elements of the subset of Pareto optimal monetary equilibria from the reference economy $\mathcal{H}^{PO}$ through approximations by elements of $\mathcal{H}^{PO}_{k}$, $k\geq1$. The next result answers affirmatively.

\begin{theorem}[Backward calculation algorithm]\label{theoAlgorithmFinal}{
If $\{\mathcal{E}_{k}\}_{k\geq1}$ is finitely replicating $\mathcal{E}$, then 
\begin{eqnarray*}
\lim_{j\rightarrow\infty}\overline{\bigcup_{k\geq j}\mathcal{H}^{PO}_{k}}\subseteq \mathcal{H}^{PO},
\end{eqnarray*}
with $\lim_{j\rightarrow\infty}\overline{\bigcup_{k\geq j}\mathcal{H}^{PO}_{k}}$ a non-empty compact set.}   
\end{theorem}

Theorem \ref{theoAlgorithmFinal} states that the backward calculation algorithm (i.e., starting from the well-known sets $\mathcal{H}^{PO}_{\tau k}$, $k\geq1$, moving backward through equilibrium equations to calculate $\mathcal{H}^{PO}_{k}$ and finding an accumulation point) leads to at least one element of $\mathcal{H}^{PO}$.

The distinctive features of this algorithm compared to the classical ones (e.g., \textcite[p. 238-242]{KehoeLevine_1990}) are: (i) the initial period is a sort of ``loose end'' (since there are no market clearing equations at $t=0$); (ii) the terminal (i.e., boundary) conditions are set-valued, given by the well-known sets $\{\mathcal{H}^{PO}_{\tau k}\}_{k\geq1}$; (iii) there is a clear backward direction driving the period-by-period solution of equilibrium equations that simplifies it; and (iv) the result is not just an equilibrium of $\mathcal{E}$, but an efficient (not necessarily unique) one. Furthermore, although the statement of Theorem \ref{theoAlgorithmFinal} does not aim at classical overlapping generations economies with fiat money $\mathcal{E}^{2}$, an analogous result holds due to Theorem \ref{theoClassModelExistence}. 

The following example illustrates the backward calculation algorithm.

\begin{example}\label{ex1}
The overlapping generations economy $\mathcal{E}$ is stationary. $L_{t}=1$, $t\geq0$, and households $h\in G_{t}$, $t\geq0$, have a common utility function $u^{h}(c_{t},c_{t+1})=u(c_{t},c_{t+1})=\sqrt{c_{t}}+\sqrt{3c_{t+1}/5}$ and a common endowment $e^{h}=e=(4,4/5)$. Assumptions \ref{assBoundsPopAndEndownments}--\ref{assCassCriterion} are trivially satisfied. Equilibrium equations are given by
\begin{eqnarray*}
\sum_{h\in G_{t-1}}(x^{h}_{t}(p_{t-1},p_{t})-e^{h}_{t})+\sum_{h\in G_{t}}(x^{h}_{t}(p_{t},p_{t+1})-e^{h}_{t})=0,
\end{eqnarray*}
for $t\geq1$. We normalize $p_{0}=1$ and let $r_{t}=p_{t}/p_{t+1}$, $t\geq0$. Excess demand of household $h\in G_{t}$, $t\geq0$, is given by  
\begin{eqnarray*}
(x^{h}_{t}(r_{t})-e^{h}_{t},x^{h}_{t+1}(r_{t})-e^{h}_{t+1})= (-\phi(r_{t}),r_{t}\phi(r_{t})),
\end{eqnarray*}
with $\phi:\mathbb{R}_{++}\rightarrow\mathbb{R}$ given by
\begin{eqnarray*}
    \phi(r)=\frac{12r^{2}-4}{3r^{2}+5r}=4\biggr(1-\frac{1+5r}{3r^{2}+5r}\biggr).
\end{eqnarray*}
Then, equilibrium equations can be conveniently written as
\begin{eqnarray*}
\phi(r_{t})=\sum_{h\in G_{t}}(e^{h}_{t}-x^{h}_{t}(r_{t}))=\sum_{h\in G_{t-1}}(x^{h}_{t}(r_{t-1})-e^{h}_{t})=r_{t-1}\phi(r_{t-1}),    
\end{eqnarray*}
for $t\geq1$. In order to apply Theorem \ref{theoAlgorithmFinal}, we need to first show that the economy $\mathcal{E}$ satisfies Assumption \ref{assProneSavingsEconomy}. Notice that for $h\in G_{t}$, $t\geq0$, $s^{h}(p_{t},p_{t+1})=\phi(r_{t})$ and $\phi(\cdot)$ is a strictly increasing continuous function, with $\lim_{r\rightarrow0^{+}}\phi(r)=-\infty$, $\lim_{r\rightarrow+\infty}\phi(r)=4$, and $\phi(\sqrt{3}/3)=0$. Therefore, there are $\varepsilon>0$ and $\delta>0$ such that $\phi((1+\varepsilon)^{-1})=\delta$. Then,
\begin{eqnarray*}
    \sum_{h\in G_{t}}\frac{s^{h}(p_{t},p_{t+1})}{L_{t}}=\phi(r_{t})\leq\delta=\phi\biggr(\frac{1}{1+\varepsilon}\biggr)\implies r_{t}\leq\frac{1}{1+\varepsilon}=\frac{\alpha_{t}}{1+\varepsilon},
\end{eqnarray*}
for $r_{t}\geq0$, $t\geq0$, and the economy is prone to savings. 

Through the results of \textcite{Gale_1973} and after some thought on the construction of a convenient finitely replicating sequence $\{\mathcal{E}_{k}\}_{k\geq0}$, notice that it is always possible, for any given $w>0$ with $\phi(w)>0$, to choose tail economies $\{\mathcal{T}_{k}\}_{k\geq1}$ for which the set of Pareto optimal monetary equilibria is uniquely defined through the real savings of generation $G^{\tau k}_{0}$ equating $\phi(w)$. With this sort of finitely replicating economies, the backward calculation algorithm from Theorem \ref{theoAlgorithmFinal} requires to solve
\begin{eqnarray*}
    \phi(r^{k}_{1})&=&r^{k}_{0}\phi(r^{k}_{0})\\
    \phi(r^{k}_{2})&=&r^{k}_{1}\phi(r^{k}_{1})\\
    &\vdots&\\
    \phi(w)&=&r^{k}_{k}\phi(r^{k}_{k}),
\end{eqnarray*}
in order to find the real return rates $(r^{k}_{0},\ldots,r^{k}_{k})\in\mathbb{R}^{k+1}_{++}$ that define $\mathcal{H}^{PO}_{k}$ on the first $k+1$ periods. Next, let $\psi:(-4/5,+\infty)\rightarrow \mathbb{R}_{++}$ be given by
\begin{eqnarray*}
    \psi(y)=\frac{3y+\sqrt{9y^{2}+240y+192}}{24},
\end{eqnarray*}
so that $\psi(r\phi(r))=r$, $r>0$. The solution of the equilibrium equations can be written as
\begin{eqnarray}\label{eq1Ex2}
    r^{k}_{t}=[\psi\circ\phi]^{k-t+1}(w),
\end{eqnarray}
for $0\leq t\leq k$. Since in overlapping generations economies with two-periods-living households and a single good per period an equilibrium is fully defined by any of its real return rates, our aim is to calculate $r_{0}>0$ that yields an element of $\mathcal{H}^{PO}$. Notice that (\ref{eq1Ex2}) implies
\begin{eqnarray*}
    p^{k}_{1}=\frac{p^{k}_{0}}{r^{k}_{0}}=\frac{1}{r^{k}_{0}}=\frac{1}{[\psi\circ\phi]^{k+1}(w)},
\end{eqnarray*}
with $p^{k}_{1}=\pi_{1}(\mathcal{H}^{PO}_{k})$, $k\geq1$, and $\pi_{1}(\cdot)$ the respective canonical projection. Theorem \ref{theoAlgorithmFinal} states that 
\begin{eqnarray*}
    \lim_{j\rightarrow\infty}\overline{\bigcup_{k\geq j}\mathcal{H}^{PO}_{k}}\subseteq \mathcal{H}^{PO},
\end{eqnarray*}
and, since $\overline{\bigcup_{k\geq j}\mathcal{H}^{PO}_{k}}\subseteq\mathcal{K}^{\prime}$, $j\geq1$, is compact, we have 
\begin{eqnarray*}
     \pi_{1}\biggr(\lim_{j\rightarrow\infty}\overline{\bigcup_{k\geq j}\mathcal{H}^{PO}_{k}}\biggr)=\lim_{j\rightarrow\infty}\pi_{1}\biggr(\overline{\bigcup_{k\geq j}\mathcal{H}^{PO}_{k}}\biggr)=\lim_{j\rightarrow\infty}\overline{\bigcup_{k\geq j}\biggr\{\frac{1}{[\psi\circ\phi]^{k+1}(w)}\biggr\}}\subseteq \pi_{1}(\mathcal{H}^{PO}).
\end{eqnarray*}
Next, notice that $\psi\circ\phi:\mathbb{R}_{++}\rightarrow \mathbb{R}_{++}$ is a strictly concave, strictly increasing function, with $\lim_{r\rightarrow0^{+}}\psi\circ\phi(r)=0$ and $\psi\circ\phi(1)=1$. Therefore, for all $w>0$, $\lim_{k\rightarrow\infty}[\psi\circ\phi]^{k+1}(w)=1$ and
\begin{eqnarray*}
    \overline{\bigcup_{k\geq j}\biggr\{\frac{1}{[\psi\circ\phi]^{k+1}(w)}\biggr\}}&=&\biggr\{\bigcup_{k\geq j}\biggr\{\frac{1}{[\psi\circ\phi]^{k+1}(w)}\biggr\}\biggr\}\cup \{1\}\implies\\
    \lim_{j\rightarrow\infty}\overline{\bigcup_{k\geq j}\biggr\{\frac{1}{[\psi\circ\phi]^{k+1}(w)}\biggr\}} &=& \lim_{j\rightarrow\infty}\biggr\{\bigcup_{k\geq j}\biggr\{\frac{1}{[\psi\circ\phi]^{k+1}(w)}\biggr\}\biggr\}\cup \{1\}=\{1\}.
\end{eqnarray*}

Finally, solving forward the equilibrium equations after $r_{0}=1$ yields $p=(1,1,\ldots)\in\mathcal{H}^{PO}$.
\end{example}

There are still three remarks to be made about Theorem \ref{theoAlgorithmFinal} and the backward calculation algorithm.

The first and, I believe, most \textit{fundamental} remark is that Theorem \ref{theoAlgorithmFinal} provides a justification for working with overlapping generations economies that possess well-behaved tails to model monetary policy, social security systems, etc. Let me be more specific. Even if the model maker does not know precisely how generations are going to evolve a decade or half a century from now, Theorem \ref{theoAlgorithmFinal} states that working with a well-behaved \textit{far-enough tail} (i.e., working with $\mathcal{E}_{k}$, $k>>1$) and calculating backward the equilibria departing from the set of Pareto optimal monetary equilibria of this tail leads, approximately, to efficient monetary equilibria of the ``best'' possible model economy $\mathcal{E}$.  

Clearly, the approximation holds in a topological sense and will vary according to the chosen far-enough tail. Nevertheless, Theorem \ref{theoAlgorithmFinal} provides a theoretical basis for what is actually done in many economic models (see, for example, two-stage dividend discount models or actuarial evaluations that, after a given year, simply assume a constant growth rate of the economy). 

Second, other than an algorithm for calculating Pareto optimal monetary equilibria, Theorem \ref{theoAlgorithmFinal} can be seen as a loosening of the perfect foresight hypothesis, since it states that today's economic agents can ``simplify'' the tail of the economy and not lead the current economic landscape to diverge substantially from perfect foresight efficient equilibria.

A related matter is that even if today's economic agents do not know exactly how the future will be and simplify the tail of the economy, Theorem \ref{theoAlgorithmFinal} and Proposition \ref{propNecAndSuffConditionParetoOpt} allow us to state that the key macroeconomic variable to be observed in order to monitor the optimality of the current economic landscape is real savings per capita. If real savings per capita start to plummet, we enter a dangerous zone.

The third and final remark is the following. Looking at the dimensions of the consumption spaces in each period, the tail economies $\{\mathcal{T}_{k}\}_{k\geq1}$ with the lowest possible dimensions are those in which $L^{\tau k}_{0}=L_{k+1}$ and $L^{\tau k}_{t}=1$, $t\geq1$. These low dimensional tail economies facilitate the backward calculation given by Theorem \ref{theoAlgorithmFinal}, but one might suspect that they are too meager and therefore leave aside values of Pareto optimal monetary equilibria that one would like to find. This intuition can be made rigorous in the following sense. 

A well-known result by \textcite{Debreu_1970}, following Sard's theorem, allows us to state that, generically, $\mathbf{J}Z_{j}(p)$, $j\geq1$, is surjective for all $p\in\mathcal{Z}_{j}$ and, therefore, $\mathcal{Z}_{j}$ is a smooth manifold without border of dimension $L_{0}+L_{j+1}-1$ \parencite{GuilleminPollack_1974} with $\mathcal{V}_{j}\subseteq\mathcal{Z}_{j}$.

Theorem \ref{theoEqSetCompact} implies that if $p\in\mathcal{H}^{PO}$ then $\pi_{t+1}(p)\in\pi_{t+1}(\mathcal{V}_{j})$, for all $j\geq t$, with $\pi_{t+1}(p)=(p_{0},\ldots,p_{t+1})$. Therefore, all prices until period $t+1$ of the Pareto optimal monetary equilibria of $\mathcal{E}$ are contained in $\pi_{t+1}(\mathcal{V}_{j})\subseteq\pi_{t+1}(\mathcal{Z}_{j})$, $j\geq t$. In other words, all finite components $(p_{0},\ldots,p_{t+1})$ of $p\in\mathcal{H}^{PO}$ are contained in the projection of manifolds of dimension $L_{0}+L_{j+1}-1$, $j\geq t$.

However, when applying the backward calculation algorithm from Theorem \ref{theoAlgorithmFinal}, we aim to approximate elements from $\mathcal{H}^{PO}$ by elements from $\mathcal{H}^{PO}_{k}$, $k\geq1$.  Using the same reasoning as before, all the finite components $(p^{k}_{0},\ldots,p^{k}_{t+1})$ of $p^{k}\in\mathcal{H}^{PO}_{k}$ are contained in the projection of manifolds of dimension $L^{k}_{0}+L^{k}_{j+1}-1=L_{0}+L^{k}_{j+1}-1$, $j\geq t$. 

If $L^{k}_{j}=1$, $j\geq k+2$, we see that all finite components $(p^{k}_{0},\ldots,p^{k}_{t+1})$ of $p^{k}\in\mathcal{H}^{PO}_{k}$ are contained in the projection of manifolds of dimension $L_{0}+L^{k}_{j+1}-1=L_{0}$, $j\geq t+1$. Since $L_{0}\leq L_{0}+L_{j+1}-1$, $j\geq t$, we are, if $\inf_{t\geq0} \{L_{t}\}>1$, ``losing'' dimensions relative to the reference economy $\mathcal{E}$. This remark reveals, in particular, that when choosing tail economies $\{\mathcal{T}_{k}\}_{k\geq1}$ in order to build finitely replicating sequences $\{\mathcal{E}_{k}\}_{k\geq1}$ of $\mathcal{E}$ and apply the algorithm of Theorem \ref{theoAlgorithmFinal}, one must choose dimensions such that $L^{\tau k}_{t}\geq \inf_{i\geq0} \{L_{i}\}$, $t\geq0$, $k\geq1$.

\section{Concluding remarks}\label{sec5}

From a theoretical perspective, this paper provides sufficient conditions for the existence of Pareto optimal monetary equilibria on consumption-loan, non-stationary overlapping generations economies. Fundamentally, the economies must be prone to savings, a condition that relates real savings per capita to the real return rate of the economy and can be linked to skewed aggregate endowment distributions toward younger periods of life, sufficiently high demographic growth rates, and meaningfully valued consumption in older periods of life. 

Due to its macroeconomic implications, a result from \hyperref[sec3]{Section 3} that deserves careful attention is Proposition \ref{propNecAndSuffConditionParetoOpt}, which shows that the dynamics of real savings per capita completely characterize efficient equilibria in prone-to-savings economies. When it comes to monetary policy, this indicates that it is important to monitor not only the price level (since inflation can make (\ref{eqEfficientEq}) converge and, therefore, lead to an inefficient allocation) but also real savings per capita.

From an applied perspective, this paper shows that an adequate way to calculate efficient monetary equilibria in non-stationary overlapping generations economies (and, in particular, transition paths) is by choosing convenient tail economies with well-known sets of equilibria and then using market clearing equations to move backward. This algorithm provides a clear step-by-step direction for solving equilibrium equations and approximating efficient monetary equilibria.

Finally, the specific economic problem that motivated the results of this paper should be clearly stated. Worldwide plummeting fertility rates are leading to a rapid demographic transition, putting high pressure on pay-as-you-go social security systems. The optimal design of these systems can be achieved through overlapping generations models and, specifically, through non-stationary consumption-loan models with heterogeneous households \parencite{Dognini_2026}. However, this design requires the calculation of at least one efficient equilibrium.

\appendix
\section*{Appendix}\label{appn}
All proofs from the results in \hyperref[sec3]{Section 3} and \hyperref[sec4]{Section 4} are given in this \hyperref[appn]{Appendix}.
\begin{proof}[Proof of Lemma \ref{lemmaVjInfinite}]
First, we must prove that $\mathcal{V}_{j}\subset\mathbb{R}^{\sum^{j+1}_{i=0}L_{i}}_{++}$, $j\geq1$, is compact. Let $(p_{0},\ldots,p_{j+1})\in\mathcal{V}_{j}$, $j\geq1$. By Definition \ref{defSetJSightedEquilibria}, $p_{01}=1$, $(p_{0},p_{1})\in\mathcal{B}_{0}(\sigma_{0})$ and $(p_{j}/p_{j1},p_{j+1}/p_{j1})\in\mathcal{B}_{j}(\sigma_{j})$. Market clearing equations at $1\leq t \leq j$ imply the following bound on average demand of generation $G_{t}$ when young,
\begin{eqnarray*}
\biggr\Vert \sum_{h\in G_{t}} \frac{x^{h}_{t}(p_{t},p_{t+1})}{H_{t}}\biggr\Vert_{\infty}&=&\biggr\Vert \sum_{h\in G_{t}} \frac{e^{h}_{t}}{H_{t}} + \frac{H_{t-1}}{H_{t}}\sum_{h\in G_{t-1}}\frac{e^{h}_{t}-x^{h}_{t}(p_{t-1},p_{t})}{H_{t-1}}\biggr\Vert_{\infty}\\
&\leq&\biggr\Vert \sum_{h\in G_{t}} \frac{e^{h}_{t}}{H_{t}} \biggr\Vert_{\infty}+\alpha^{-1}_{\min}\biggr\Vert\sum_{h\in G_{t-1}} \frac{e^{h}_{t}}{H_{t-1}}\biggr\Vert_{\infty}\\
&\leq& (1+\alpha^{-1}_{\min})e_{\max}\leq \beta e_{\max}. 
\end{eqnarray*}

Also, market clearing at $1\leq t \leq j$ implies the following bound on average demand of generation $G_{t-1}$ when old,
\begin{eqnarray*}
\biggr\Vert \sum_{h\in G_{t-1}} \frac{x^{h}_{t}(p_{t-1},p_{t})}{H_{t-1}}\biggr\Vert_{\infty}&=&\biggr\Vert \sum_{h\in G_{t-1}} \frac{e^{h}_{t}}{H_{t-1}} + \frac{H_{t}}{H_{t-1}}\sum_{h\in G_{t}}\frac{e^{h}_{t}-x^{h}_{t}(p_{t},p_{t+1})}{H_{t}}\biggr\Vert_{\infty}\\
&\leq&\biggr\Vert \sum_{h\in G_{t-1}} \frac{e^{h}_{t}}{H_{t-1}} \biggr\Vert_{\infty}+\alpha_{\max}\biggr\Vert\sum_{h\in G_{t}} \frac{e^{h}_{t}}{H_{t}}\biggr\Vert_{\infty}\\
&\leq& (1+\alpha_{\max})e_{\max}\leq \beta e_{\max}.
\end{eqnarray*}
We conclude that average demand of generation $G_{t}$, $1\leq t\leq j-1$, satisfies
\begin{eqnarray*}
   \biggr\Vert \sum_{h\in G_{t}} \frac{x^{h}(p_{t},p_{t+1})}{H_{t}}\biggr\Vert_{\infty}\leq \beta e_{\max},
\end{eqnarray*}
and Assumption \ref{assBoundsPrices} implies that $(p_{t}/p_{t1},p_{t+1}/p_{t1})\in\mathcal{B}_{t}(\sigma_{t})$, $1\leq t\leq j-1$. Therefore,  $(p_{t}/p_{t1},p_{t+1}/p_{t1})\in\mathcal{B}_{t}(\sigma_{t})$, for all $0\leq t\leq j$. Since $(p_{0},p_{1})\in\mathcal{B}_{0}(\sigma_{0})$, then
\begin{eqnarray*}
    (\sigma_{0},\ldots,\sigma_{0})\leq p_{0} \leq  (\sigma_{0}^{-1},\ldots,\sigma_{0}^{-1})\\
    (\sigma_{0},\ldots,\sigma_{0})\leq p_{1} \leq  (\sigma_{0}^{-1},\ldots,\sigma_{0}^{-1}),
\end{eqnarray*}
so that $p_{0},p_{1}\in\mathcal{K}_{0}=\mathcal{K}_{1}$. In particular, $\sigma_{0}\leq p_{11}\leq \sigma_{0}^{-1}$. Then, $(p_{1}/p_{11},p_{2}/p_{11})\in\mathcal{B}_{1}(\sigma_{1})$ implies
\begin{eqnarray*}
     (\sigma_{0}\sigma_{1},\ldots,\sigma_{0}\sigma_{1})\leq 
     (\sigma_{1} p_{11},\ldots,\sigma_{1} p_{11})\leq p_{2} \leq  \biggr(\frac{p_{11}}{\sigma_{1}},\ldots,\frac{p_{11}}{\sigma_{1}}\biggr)\leq \biggr(\frac{1}{\sigma_{0}\sigma_{1}},\ldots,\frac{1}{\sigma_{0}\sigma_{1}}\biggr).
\end{eqnarray*}
By induction, we can assert that 
\begin{eqnarray*}
    \biggr(\prod^{t-1}_{j=0}\sigma_{j},\ldots,\prod^{t-1}_{j=0}\sigma_{j}\biggr)\leq p_{t} \leq \biggr(\frac{1}{\prod^{t-1}_{j=0}\sigma_{j}},\ldots,\frac{1}{\prod^{t-1}_{j=0}\sigma_{j}}\biggr),
\end{eqnarray*}
for $1\leq t\leq j+1$. We conclude that $p_{t}\in\mathcal{K}_{t}$, $0\leq t\leq j+1$, and therefore $\mathcal{V}_{j}\subseteq\prod^{j+1}_{t=0}\mathcal{K}_{t}$.

Next, let $\{p^{n}\}_{n\geq1}$ be any sequence with $p^{n}\in\mathcal{V}_{j}$, $n\geq1$, and $\lim_{n\rightarrow\infty}p^{n}=p\in\mathbb{R}^{\sum^{j+1}_{i=0}L_{i}}$. In particular, $p_{01}=\lim_{n\rightarrow\infty} p^{n}_{01}=1$. Also, for $0\leq t\leq j+1$, $p^{n}_{t}\in\mathcal{K}_{t}$, $n\geq1$, which implies that $p_{t}=\lim_{n\rightarrow\infty}p^{n}_{t}\in\mathcal{K}_{t}\subset\mathbb{R}^{L_{t}}_{++}$, since $\mathcal{K}_{t}$ is closed. Then, continuity of Walrasian demand functions on strictly positive prices leads to 
\begin{eqnarray*}
z_{t}(p_{t-1},p_{t},p_{t+1})=\lim_{n\rightarrow\infty} z_{t}(p^{n}_{t-1},p^{n}_{t},p^{n}_{t+1})=0,
\end{eqnarray*}
for $1\leq t\leq j$. Since $\mathcal{B}_{0}(\sigma_{0})$ and $\mathcal{B}_{j}(\sigma_{j})$ are closed and $\lim_{n\rightarrow\infty}p^{n}_{j1}=p_{j1}>0$, we also have $(p_{0},p_{1})\in \mathcal{B}_{0}(\sigma_{0})$ and  $(p_{j}/p_{j1},p_{j+1}/p_{j1})\in \mathcal{B}_{j}(\sigma_{j})$. Therefore, $p\in\mathcal{V}_{j}$ and $\mathcal{V}_{j}$ is a compact set. 

The fact that $\mathcal{V}_{j}$ is non-empty is derived after Theorem 5 from \textcite[119]{ArrowHahn_1971}. By Assumption \ref{assResourceRelated}, every household $h\in\bigcup^{j}_{t=0}G_{t}$ in the ``insulated'' economy formed by all generations until time $j$ is indirectly resource related to every other. Therefore, the aforementioned theorem and Assumption \ref{assBoundsPrices} imply that this economy has an equilibrium $p^{\prime}=(p^{\prime}_{0},\ldots,p^{\prime}_{j+1})\in\mathbb{R}^{\sum^{j+1}_{i=0}L_{i}}_{++}$, with $p^{\prime}_{01}=1$ after a suitable normalization. Notice that the market clearing equation for this insulated economy in period $t=0$ is given by
\begin{eqnarray*}
    \sum_{h\in G_{0}}x^{h}_{0}(p^{\prime}_{0},p^{\prime}_{1})=\sum_{h\in G_{0}}e^{h}_{0},
\end{eqnarray*}
and in period $t=j+1$ is given by
\begin{eqnarray*}
    \sum_{h\in G_{j}}x^{h}_{j+1}(p^{\prime}_{j},p^{\prime}_{j+1})=\sum_{h\in G_{j}}e^{h}_{j+1}.
\end{eqnarray*}
Then, 
\begin{eqnarray*}
\biggr\Vert \sum_{h\in G_{0}} \frac{x^{h}_{0}(p^{\prime}_{0},p^{\prime}_{1})}{H_{0}}\biggr\Vert_{\infty}=\biggr\Vert \sum_{h\in G_{0}} \frac{e^{h}_{0}}{H_{0}}\biggr\Vert_{\infty} \leq e_{\max}< \beta e_{\max},
\end{eqnarray*}
and
\begin{eqnarray*} 
\biggr\Vert \sum_{h\in G_{j}} \frac{x^{h}_{j+1}(p^{\prime}_{j},p^{\prime}_{j+1})}{H_{j}}\biggr\Vert_{\infty}=\biggr\Vert \sum_{h\in G_{j}} \frac{e^{h}_{j+1}}{H_{j}}\biggr\Vert_{\infty}\leq e_{\max}<\beta e_{\max}.  
\end{eqnarray*}
Also, the market clearing equation in period $1\leq t\leq j$ is given by $z_{t}(p^{\prime}_{t-1},p^{\prime}_{t},p^{\prime}_{t+1})=0$, and allows us to limit average demand from generations $G_{0}$ in $t=1$ and $G_{j}$ in $t=j$ by $\beta e_{\max}$. Therefore, Assumption \ref{assBoundsPrices} allows us to conclude that $p^{\prime}\in\mathcal{V}_{j}$ and $\mathcal{V}_{j}$ is non-empty.

Finally, notice that if $p\in\mathcal{V}^{\infty}_{j+1}$, then $(p_{0},p_{1})\in\mathcal{B}_{0}(\sigma_{0})$, $p_{01}=1$ and $Z_{j+1}(p_{0},\ldots,p_{j+2})=0$. In particular, $Z_{j}(p_{0},\ldots,p_{j+1})=0$ and we can, once again, use Assumption \ref{assBoundsPrices} and the market clearing equations in periods $t=j$ and $t=j+1$ to conclude that $(p_{j}/p_{j1},p_{j+1}/p_{j1})\in \mathcal{B}_{j}(\sigma_{j})$. Then, $(p_{0},\ldots,p_{j+1})\in \mathcal{V}_{j}$ and, therefore, $p\in\mathcal{V}^{\infty}_{j}$. We conclude that $\mathcal{V}^{\infty}_{j+1}\subseteq\mathcal{V}^{\infty}_{j}$.
\end{proof}

\begin{proof}[Proof of Theorem~{\upshape\ref{theoEqSetCompact}}]
Lemma \ref{lemmaVjInfinite} implies that $\{\mathcal{V}^{\infty}_{j}\cap\mathcal{K}\}_{j\geq1}$ is a nested sequence of sets, with $\mathcal{K}=\prod^{\infty}_{t=0}\mathcal{K}_{t}$ and  
\begin{eqnarray*}
    \mathcal{V}^{\infty}_{j}\cap \mathcal{K}=\mathcal{V}_{j}\times\prod^{\infty}_{t\geq j+2}\mathcal{K}_{t},
\end{eqnarray*}
for $j\geq1$. Then, Lemma \ref{lemmaVjInfinite} and Tychonoff's Theorem imply $\mathcal{V}^{\infty}_{j}\cap \mathcal{K}$, $j\geq1$, is non-empty and compact. 

First, we prove that $\mathcal{H}\subseteq\lim_{j\rightarrow\infty}(\mathcal{V}^{\infty}_{j}\cap\mathcal{K})$. Let $p\in\mathcal{H}$. Definition \ref{defSetEquilibria} implies $p_{01}=1$ and $(p_{0},p_{1})\in\mathcal{B}_{0}(\sigma_{0})$. By the same reasoning as in the proof of Lemma \ref{lemmaVjInfinite}, we can use market clearing equations to limit the average demand of generation $G_{t}$, $t\geq1$, by $\beta e_{\max}$ and through Assumption \ref{assBoundsPrices} we have
$(p_{t}/p_{t1},p_{t+1}/p_{t1})\in\mathcal{B}_{t}(\sigma_{t})$, $t\geq1$. Also, by the same reasoning, we have $p_{t}\in\mathcal{K}_{t}$, $t\geq0$, thus yielding $p\in\mathcal{K}$.

Then, $(p_{0},\ldots,p_{j+1})\in\mathcal{V}_{j}$, $j\geq1$, so that $p\in\mathcal{V}_{j}^{\infty}$, $j\geq1$. We conclude that $p\in \lim_{j\rightarrow\infty}\mathcal{V}^{\infty}_{j}\cap\mathcal{K}$ and, therefore, $\mathcal{H}\subseteq\lim_{j\rightarrow\infty}(\mathcal{V}^{\infty}_{j}\cap\mathcal{K})$. 

Next, if $p\in\lim_{j\rightarrow\infty}(\mathcal{V}^{\infty}_{j}\cap\mathcal{K})$, then $p\in \mathcal{V}^{\infty}_{j}\cap\mathcal{K}$, $j\geq1$, since $\{\mathcal{V}^{\infty}_{j}\cap\mathcal{K}\}_{j\geq1}$ is a nested sequence of sets. Definition \ref{defSetJSightedEquilibria} implies $p_{01}=1$, $(p_{0},p_{1})\in\mathcal{B}_{0}(\sigma_{0})$ and $z_{t}(p_{t-1},p_{t},p_{t+1})=0$, $t\geq1$. We conclude that $p\in\mathcal{H}$ and, therefore, $\mathcal{H}=\lim_{j\rightarrow\infty}(\mathcal{V}^{\infty}_{j}\cap\mathcal{K})$. 

Since $\mathcal{H}=\lim_{j\rightarrow\infty}(\mathcal{V}^{\infty}_{j}\cap\mathcal{K})$, Cantor's Intersection Theorem implies that $\mathcal{H}$ is non-empty and compact. Finally, notice that
\begin{eqnarray*}
    \lim_{j\rightarrow\infty}\mathcal{V}_{j}^{\infty}=\lim_{j\rightarrow\infty}(\mathcal{V}^{\infty}_{j}\cap\mathcal{K})\subseteq\mathcal{K},
\end{eqnarray*}
and, therefore, $\mathcal{H}=\lim_{j\rightarrow\infty}\mathcal{V}^{\infty}_{j}\subseteq\mathcal{K}$.
\end{proof}

\begin{proof}[Proof of Lemma~{\upshape\ref{lemmaFirstParetoOptEquilibria}}]
Since $p\in\mathcal{H}$, $z_{t}(p_{t-1},p_{t},p_{t+1})=0$, $t\geq1$, implies that
\begin{eqnarray}\label{eq1LemmaFirst}
\sum_{h\in G_{t-1}}p_{t}\cdot(x^{h}_{t}(p_{t-1},p_{t})-e^{h}_{t})+\sum_{h\in G_{t}}p_{t}\cdot(x^{h}_{t}(p_{t},p_{t+1})-e^{h}_{t})=0.
\end{eqnarray}
Summing the budget constraints of all households $h\in G_{t}$, $t\geq0$, furnishes, through Walras' law, 
\begin{eqnarray}\label{eq2LemmaFirst}
\sum_{h\in G_{t}}p_{t}\cdot(x^{h}_{t}(p_{t},p_{t+1})-e^{h}_{t})+\sum_{h\in G_{t}}p_{t+1}\cdot(x^{h}_{t+1}(p_{t},p_{t+1})-e^{h}_{t+1})=0. 
\end{eqnarray}
Then, (\ref{eq1LemmaFirst}) and (\ref{eq2LemmaFirst}) imply that 
\begin{eqnarray*}
\sum_{h\in G_{0}}p_{0}\cdot(x^{h}_{0}(p_{0},p_{1})-e^{h}_{0})=\sum_{h\in G_{t}}p_{t}\cdot(x^{h}_{t}(p_{t},p_{t+1})-e^{h}_{t}),
\end{eqnarray*}
for $t\geq1$. Next, Definition \ref{defRealSavings} implies
\begin{eqnarray*}
    \sum_{h\in G_{t}}\frac{s(p_{t},p_{t+1})}{H_{t}}=\sum_{h\in G_{t}}\frac{p_{t}}{H_{t}\Vert p_{t}\Vert}\cdot (e^{h}_{t}-x^{h}_{t}(p_{t},p_{t+1}))=\frac{1}{H_{t}\Vert p_{t}\Vert}\sum_{h\in G_{0}}p_{0}\cdot(e^{h}_{0}-x^{h}_{0}(p_{0},p_{1})).
\end{eqnarray*}
for $t\geq0$. By Assumption \ref{assCassCriterion}, 
\begin{eqnarray*}
    \mathcal{H}^{PO}=\biggr\{(p_{0},p_{1},\ldots)\in\mathcal{H}\mid \sum^{+\infty}_{t=0}\frac{1}{H_{t}\Vert p_{t}\Vert}=+\infty\biggr\}.
\end{eqnarray*}
Since all terms in $\sum^{+\infty}_{t=0}(H_{t}\Vert p_{t}\Vert)^{-1}$ are positive, $p\notin\mathcal{H}^{PO}$ implies convergence of the series and, therefore, $\lim_{t\rightarrow\infty}(H_{t}\Vert p_{t}\Vert)^{-1}=0$. Finally,
\begin{eqnarray*}
 \lim_{t\rightarrow\infty}\sum_{h\in G_{t}}\frac{s(p_{t},p_{t+1})}{H_{t}}=\lim_{t\rightarrow\infty}\frac{1}{H_{t}\Vert p_{t}\Vert}\sum_{h\in G_{0}}p_{0}\cdot(e^{h}_{0}-x^{h}_{0}(p_{0},p_{1}))=0.
\end{eqnarray*}
\end{proof}

\begin{proof}[Proof of Proposition~{\upshape\ref{propSufficientConditionProneSavings}}]
Let $t\geq0$, $(p_{t},p_{t+1})\in\mathbb{R}^{L_{t}+L_{t+1}}_{++}$ and $(p_{t}/p_{t1},p_{t+1}/p_{t1})\in\mathcal{B}_{t}(\sigma_{t})$. Real savings of household $h\in G_{t}$, $t\geq0$, is given by
\begin{eqnarray*}
s^{h}(p_{t},p_{t+1})=\frac{p_{t}\cdot e^{h}_{t}-\gamma(p_{t},p_{t+1})(p_{t},p_{t+1})\cdot e^{h}}{\Vert p_{t} \Vert},
\end{eqnarray*}
with
\begin{eqnarray*}
\gamma(p_{t},p_{t+1}) = \frac{\sum^{L_{t}}_{i=1}\lambda^{\eta_{t}}_{ti}p^{1-\eta_{t}}_{ti}}{\sum^{L_{t}}_{i=1}\lambda^{\eta_{t}}_{ti}p^{1-\eta_{t}}_{ti} + \sum^{L_{t+1}}_{i=1}\tilde{\lambda}^{\eta_{t}}_{(t+1)i}p^{1-\eta_{t}}_{(t+1)i}},  
\end{eqnarray*}
and $\eta_{t}=1/(1-\rho_{t})>1$. Real savings per capita of generation $G_{t}$, $t\geq0$, is given by
\begin{eqnarray*}
  \sum_{h\in G_{t}}\frac{s^{h}(p_{t},p_{t+1})}{H_{t}}=  \frac{p_{t}\cdot e^{t}_{t}-\gamma(p_{t},p_{t+1}) (p_{t},p_{t+1})\cdot e^{t} }{\Vert p_{t} \Vert},
\end{eqnarray*}
with $e^{t}=(e^{t}_{t},e^{t}_{t+1})=\sum_{h\in G_{t}}e^{h}/H_{t}\in\mathbb{R}^{L_{t}+L_{t+1}}_{++}$. Notice that $(p_{t}/p_{t1},p_{t+1}/p_{t1})\in\mathcal{B}_{t}(\sigma_{t})$ implies 
\begin{eqnarray*}
    (\sigma_{t} p_{t1},\ldots,\sigma_{t} p_{t1})\leq (p_{t},p_{t+1})\leq (\sigma^{-1}_{t} p_{t1},\ldots,\sigma^{-1}_{t} p_{t1}).
\end{eqnarray*}
Since $\gamma(p_{t},p_{t+1})$ is strictly decreasing on its first $L_{t}$ arguments and strictly increasing on the last $L_{t+1}$, then
\begin{eqnarray*}
\gamma(p_{t},p_{t+1})&\leq& \frac{\sigma^{1-\eta_{t}}_{t}\sum^{L_{t}}_{i=1}\lambda^{\eta_{t}}_{ti}}{\sigma^{1-\eta_{t}}_{t}\sum^{L_{t}}_{i=1}\lambda^{\eta_{t}}_{ti}+ \sigma^{\eta_{t}-1}_{t}\sum^{L_{t+1}}_{i=1}\tilde{\lambda}^{\eta_{t}}_{(t+1)i}}\\
&=&\frac{\sum^{L_{t}}_{i=1}\lambda^{\eta_{t}}_{ti}}{\sum^{L_{t}}_{i=1}\lambda^{\eta_{t}}_{ti}+ \sigma^{2(\eta_{t}-1)}_{t}\sum^{L_{t+1}}_{i=1}\tilde{\lambda}^{\eta_{t}}_{(t+1)i}}=\gamma_{t}.  
\end{eqnarray*}

Next, since $\gamma_{t}\leq 1$, $\Vert e^{t}_{t+1}\Vert_{\infty}\leq e_{\max}$ and $\alpha^{-1}_{t}\leq\alpha^{-1}_{\min}$, then
\begin{eqnarray*}
\inf_{t\geq0}\{(1-\gamma_{t})\min_{1\leq i\leq L_{t}}e^{t}_{ti}-\gamma_{t}\Vert e^{t}_{t+1}\Vert_{\infty}\alpha^{-1}_{t}\}>0,
\end{eqnarray*}
implies the existence of $\varepsilon>0$ and $\delta>0$ such that
\begin{eqnarray*}
(1-\gamma_{t})\min_{1\leq i\leq L_{t}}e^{t}_{ti}-\gamma_{t}\Vert e^{t}_{t+1}\Vert_{\infty}(1+\varepsilon)\alpha^{-1}_{t}>\delta,
\end{eqnarray*}
for $t\geq0$. Finally, if $\Vert p_{t}\Vert/\Vert p_{t+1} \Vert>\alpha_{t}/(1+\varepsilon)$, $t\geq0$, then
\begin{eqnarray*}
 \sum_{h\in G_{t}}\frac{s^{h}(p_{t},p_{t+1})}{H_{t}}&=&  (1-\gamma(p_{t},p_{t+1}))\frac{p_{t}}{\Vert p_{t} \Vert}\cdot e^{t}_{t}-\gamma(p_{t},p_{t+1}) \frac{\Vert p_{t+1}\Vert}{\Vert p_{t}\Vert} \frac{p_{t+1}}{\Vert p_{t+1}\Vert}\cdot e^{t}_{t+1} \\
 &\geq&(1-\gamma_{t})\min_{1\leq i\leq L_{t}}e^{t}_{ti}-\gamma_{t}\Vert e^{t}_{t+1}\Vert_{\infty}(1+\varepsilon)\alpha_{t}^{-1} >\delta.
\end{eqnarray*}
The contrapositive yields the proposition statement.
\end{proof}

\begin{proof}[Proof of Lemma~{\upshape\ref{lemmaSecondParetoOptEquilibria}}]
Let $p\in\mathcal{H}$ and $t\geq0$ be such that $\sum_{h\in G_{t}}s(p_{t},p_{t+1})/H_{t}\leq \delta$. If $t=0$, then Definition \ref{defSetEquilibria} implies $p_{01}=1$ and $(p_{0},p_{1})\in\mathcal{B}_{0}(\sigma_{0})$. If $t>0$, market clearing equations and Assumption \ref{assBoundsPrices} imply $(p_{t}/p_{t1},p_{t+1}/p_{t1})\in\mathcal{B}_{t}(\sigma_{t})$. Then, by Assumption \ref{assProneSavingsEconomy}, we have
\begin{eqnarray}\label{eqLemmaSecondParetoOpt}
\frac{\Vert p_{t}\Vert}{\Vert p_{t+1}\Vert} \leq \frac{\alpha_{t}}{1+\varepsilon}.
\end{eqnarray}
Since $p\in\mathcal{H}$, we can write
\begin{eqnarray*}
    \sum_{h\in G_{t+1}}\frac{s^{h}(p_{t+1},p_{t+2})}{H_{t+1}}&=& \frac{1}{H_{t+1}}\sum_{h\in G_{t+1}}\frac{p_{t+1}}{\Vert p_{t+1}\Vert}\cdot (e^{h}_{t+1}-x^{h}_{t+1}(p_{t+1},p_{t+2}))\\
    &=&\frac{1}{H_{t+1}}\sum_{h\in G_{t}}\frac{p_{t+1}}{\Vert p_{t+1}\Vert}\cdot (x^{h}_{t+1}(p_{t},p_{t+1})-e^{h}_{t+1})\\
    &=&\frac{\Vert p_{t}\Vert}{H_{t+1}\Vert p_{t+1}\Vert}\sum_{h\in G_{t}}\frac{p_{t}}{\Vert p_{t}\Vert}\cdot (e^{h}_{t}-x^{h}_{t}(p_{t},p_{t+1}))\\
    &=&\frac{\Vert p_{t}\Vert}{\alpha_{t}\Vert p_{t+1}\Vert} \sum_{h\in G_{t}}\frac{s^{h}(p_{t},p_{t+1})}{H_{t}}\leq \frac{\delta}{1+\varepsilon}<\delta,
\end{eqnarray*}
where the first equality is derived from Definition \ref{defProneSavings}; the second from $z_{t+1}(p_{t},p_{t+1},p_{t+2})=0$; the third from Walras' law applied to all households from generation $G_{t}$; and the previous to last inequality from (\ref{eqLemmaSecondParetoOpt}). Therefore, $\sum_{h\in G_{t+1}}s(p_{t+1},p_{t+2})/H_{t+1}\leq \delta$ and, by induction, $\sum_{h\in G_{t+i}}s(p_{t+i},p_{t+i+1})/H_{t+i}\leq \delta$, for all $i\geq 0$. 

Then, by Assumptions \ref{assBoundsPrices} and \ref{assProneSavingsEconomy}, $\Vert p_{t+i+1}\Vert \geq \Vert p_{t+i}\Vert(1+\varepsilon)\alpha^{-1}_{t+i}$, $i\geq 0$, and
\begin{eqnarray*}
\Vert p_{t+i}\Vert \geq  \Vert p_{t}\Vert(1+\varepsilon)^{i}\prod^{i-1}_{j=0}\alpha^{-1}_{t+j}=\Vert p_{t}\Vert(1+\varepsilon)^{i}\frac{H_{t}}{H_{t+i}}\implies \frac{1}{H_{t+i}\Vert p_{t+i} \Vert}\leq \frac{1}{H_{t}\Vert p_{t} \Vert(1+\varepsilon)^{i}},
\end{eqnarray*}
for $i\geq 1$. Therefore, 
\begin{eqnarray*}
\sum^{+\infty}_{i=0}\frac{1}{H_{i}\Vert p_{i} \Vert}\leq\sum^{t-1}_{i=0}\frac{1}{H_{i}\Vert p_{i}\Vert}+\frac{1}{H_{t}\Vert p_{t} \Vert}\sum^{+\infty}_{i=0}\frac{1}{(1+\varepsilon)^{i}}<+\infty.
\end{eqnarray*}
By Assumption \ref{assCassCriterion}, we conclude that $p\notin\mathcal{H}^{PO}$.
\end{proof}

\begin{proof}[Proof of Proposition~{\upshape\ref{propNecAndSuffConditionParetoOpt}}]
Let $p\in\mathcal{H}$. If $p\notin\mathcal{H}^{PO}$, Lemma \ref{lemmaFirstParetoOptEquilibria} implies
\begin{eqnarray*}
\lim_{t\rightarrow\infty}\sum_{h\in G_{t}}\frac{s^{h}(p_{t},p_{t+1})}{H_{t}}=0.   
\end{eqnarray*} 
If $\lim_{t\rightarrow\infty}\sum_{h\in G_{t}}s(p_{t},p_{t+1})/H_{t}=0$, then there is $T\geq0$ such that 
\begin{eqnarray*}
\sum_{h\in G_{T}}\frac{s(p_{T},p_{T+1})}{H_{T}}\leq\delta,
\end{eqnarray*}
and, by Lemma \ref{lemmaSecondParetoOptEquilibria}, $p\notin\mathcal{H}^{PO}$.
\end{proof}

\begin{proof}[Proof of Theorem~{\upshape\ref{theoExistenceOfParetoOptJME}}]
Let $\mathcal{F}_{k}$, $k\geq0$, be the \textit{finite economy} formed by all generations $G_{t}$, $0\leq t\leq k$, from $\mathcal{E}$ and $H_{k+1}$ identical households, each with an endowment $e^{*}=(0,\ldots,0)\times (e_{\max},\ldots,e_{\max})\in\mathbb{R}^{L_{0}\times L_{k+1}}_{+}$ and utility function $u^{*}:\mathbb{R}^{L_{0}\times L_{k+1}}_{+}\rightarrow\mathbb{R}$, 
\begin{eqnarray*}
    u^{*}(c_{0},c_{k+1})=\sum^{L_{0}}_{i=1}\log c_{0i}.
\end{eqnarray*}
These ``star'' households can be seen as time travelers, since they hold an endowment in period $k+1$ but only consume in period $0$ (i.e., they trade with their parents' generation and with their $(k-1)$-great-grandparents' generation). Furthermore, these time travelers provide a way for ``chopping off a finite segment of the infinite model and then tying the two ends together to form a closed loop'' \parencite[p. 365]{CassYaari_1966}.

Notice that $u^{*}(\cdot)$ is continuous, non-decreasing, semi-strictly quasiconcave and without local maxima. Also, every time traveler is resource related to at least one household from generation $G_{k}$ and, therefore, Assumption \ref{assResourceRelated} implies that all households from $\mathcal{F}_{k}$ are indirectly resource related.

Then, Theorem 5 from \textcite[p. 119]{ArrowHahn_1971} and Assumption \ref{assBoundsPrices} imply the existence of an equilibrium $p^{k}=(p^{k}_{0},\ldots,p^{k}_{k+1})\in\mathbb{R}^{\sum^{k+1}_{i=0}L_{i}}_{++}$ for $\mathcal{F}_{k}$, $k\geq0$, with $p^{k}_{01}=1$ after a suitable normalization. Notice that the market clearing equation in period $t=0$ and Assumption \ref{assBoundsPopAndEndownments} allow us to write
\begin{eqnarray}
\biggr\Vert \sum_{h\in G_{0}} \frac{x^{h}_{0}(p^{k}_{0},p^{k}_{1})}{H_{0}}\biggr\Vert_{\infty}&\leq& \biggr\Vert \sum_{h\in G_{0}} \frac{x^{h}_{0}(p^{k}_{0},p^{k}_{1})}{H_{0}}+ \frac{H_{k+1} x^{*}_{0}(p^{k}_{0},p^{k}_{k+1})}{H_{0}}\biggr\Vert_{\infty}\nonumber\\
&=&\biggr\Vert \sum_{h\in G_{0}} \frac{e^{h}_{0}}{H_{0}}\biggr\Vert_{\infty}\leq e_{\max} < \beta e_{\max}.\label{eqExist1}
\end{eqnarray}
The market clearing equation in period $t=k+1$ and Assumption \ref{assBoundsPopAndEndownments} also allow us to write
\begin{eqnarray}
\biggr\Vert \sum_{h\in G_{k}} \frac{x^{h}_{k+1}(p^{k}_{k},p^{k}_{k+1})}{H_{k}}\biggr\Vert_{\infty}&=& \biggr\Vert \sum_{h\in G_{k}} \frac{x^{h}_{k+1}(p^{k}_{k},p^{k}_{k+1})}{H_{k}}+ \frac{H_{k+1} x^{*}_{k+1}(p^{k}_{0},p^{k}_{k+1})}{H_{k}}\biggr\Vert_{\infty}\nonumber\\
&=&\biggr\Vert \sum_{h\in G_{k}} \frac{e^{h}_{k+1}}{H_{k}}+\frac{H_{k+1}}{H_{k}}(e_{\max},\ldots,e_{\max})\biggr\Vert_{\infty}\nonumber\\
&\leq& (1+\alpha_{\max}) e_{\max} \leq  \beta e_{\max},\label{eqExist2}
\end{eqnarray}
where the first equality is due to the fact that $x^{*}_{k+1}(p^{k}_{0},p^{k}_{k+1})=0$. Furthermore, market clearing in $t=1$ and $t=k$, (\ref{eqExist1}), (\ref{eqExist2}) and Assumption \ref{assBoundsPrices} imply $(p^{k}_{0},p^{k}_{1})\in \mathcal{B}_{0}(\sigma_{0})$ and $(p^{k}_{k}/p^{k}_{k1},p^{k}_{k+1}/p^{k}_{k1})\in \mathcal{B}_{k}(\sigma_{k})$.

Therefore, reasoning, once again, as in the proof of Lemma \ref{lemmaVjInfinite}, we conclude that $p^{k}_{t}\in\mathcal{K}_{t}$, for all $t\leq k+1$. Let $\tilde{p}^{\,k}\in \prod_{t\geq0}\mathcal{K}_{t}$ be an extension of $p^{k}$, $k\geq0$, to $\mathbb{R}^{\infty}$ (i.e., $\tilde{p}^{\,k}_{t}=p^{k}_{t}$, $0\leq t\leq k+1$). Since $\prod_{t\geq0}\mathcal{K}_{t}$ is compact, let $\{\tilde{p}^{\,k_{n}}\}_{n\geq1}$ be a convergent subsequence, with $p=\lim_{n\rightarrow\infty}\tilde{p}^{\,k_{n}}$.

Then, $p^{k}_{01}=1$ and $(p^{k}_{0},p^{k}_{1})\in\mathcal{B}_{0}(\sigma_{0})$, $k\geq0$, imply $p_{01}=1$ and $(p_{0},p_{1})\in\mathcal{B}_{0}(\sigma_{0})$. Also, for every $t\geq1$, there is $N\geq1$ such that $n\geq N$ implies $k_{n}\geq t$. The market clearing equation in period $t\geq1$ for the finite economy $\mathcal{F}_{k_{n}}$, $n\geq N$, implies 
\begin{eqnarray*}
    z_{t}(\tilde{p}^{\,k_{n}}_{t-1},\tilde{p}^{\,k_{n}}_{t},\tilde{p}^{\,k_{n}}_{t+1})=z_{t}(p^{k_{n}}_{t-1},p^{k_{n}}_{t},p^{k_{n}}_{t+1})=0,
\end{eqnarray*}
and, therefore, the continuity of $z_{t}(\cdot)$, $t\geq1$, over strictly positive prices implies
\begin{eqnarray*}
     z_{t}(p_{t-1},p_{t},p_{t+1})=\lim_{n\rightarrow\infty} z_{t}(\tilde{p}^{\,k_{n}}_{t-1},\tilde{p}^{\,k_{n}}_{t},\tilde{p}^{\,k_{n}}_{t+1})=0.
\end{eqnarray*}
We conclude that $p\in\mathcal{H}$. Suppose that $p\notin\mathcal{H}^{PO}$. Lemma \ref{lemmaFirstParetoOptEquilibria} implies 
\begin{eqnarray*}
    \lim_{t\rightarrow\infty}\sum_{h\in G_{t}}\frac{s^{h}(p_{t},p_{t+1})}{H_{t}}=0,
\end{eqnarray*}
and, therefore, there is $T\geq1$ such that
\begin{eqnarray*}
    \sum_{h\in G_{T}}\frac{s^{h}(p_{T},p_{T+1})}{H_{T}}<\frac{\delta}{2}.
\end{eqnarray*}
Convergence on the product topology implies that there is $m\geq 1$ such that $k_{m}\geq T$ and 
\begin{eqnarray}\label{eqExistSub}
    \sum_{h\in G_{T}}\frac{s^{h}(\tilde{p}^{\,k_{m}}_{T},\tilde{p}^{\,k_{m}}_{T+1})}{H_{T}}=\sum_{h\in G_{T}}\frac{s^{h}(p^{k_{m}}_{T},p^{k_{m}}_{T+1})}{H_{T}}\leq\delta.
\end{eqnarray}
If $k_{m}=T$, then (\ref{eqExistSub}) implies 
\begin{eqnarray}\label{eqExist3}
    \sum_{h\in G_{k_{m}}}\frac{s^{h}(p^{k_{m}}_{k_{m}},p^{k_{m}}_{k_{m}+1})}{H_{k_{m}}}\leq\delta.
\end{eqnarray}
If $k_{m}>T$, we can write
\begin{eqnarray*}
    \sum_{h\in G_{T+1}}\frac{s^{h}(p^{k_{m}}_{T+1},p^{k_{m}}_{T+2})}{H_{T+1}}&=& \frac{1}{H_{T+1}}\sum_{h\in G_{T+1}}\frac{p^{k_{m}}_{T+1}}{\Vert p^{k_{m}}_{T+1}\Vert}\cdot (e^{h}_{T+1}-x^{h}_{T+1}(p^{k_{m}}_{T+1},p^{k_{m}}_{T+2}))\\
    &=&\frac{1}{H_{T+1}}\sum_{h\in G_{T}}\frac{p^{k_{m}}_{T+1}}{\Vert p^{k_{m}}_{T+1}\Vert}\cdot (x^{h}_{T+1}(p^{k_{m}}_{T},p^{k_{m}}_{T+1})-e^{h}_{T+1})\\
    &=&\frac{\Vert p^{k_{m}}_{T}\Vert}{H_{T+1}\Vert p^{k_{m}}_{T+1}\Vert}\sum_{h\in G_{T}}\frac{p^{k_{m}}_{T}}{\Vert p^{k_{m}}_{T}\Vert}\cdot (e^{h}_{T}-x^{h}_{T}(p^{k_{m}}_{T},p^{k_{m}}_{T+1}))\\
    &=&\frac{\Vert p^{k_{m}}_{T}\Vert}{\alpha_{T}\Vert p^{k_{m}}_{T+1}\Vert} \sum_{h\in G_{T}}\frac{s^{h}(p^{k_{m}}_{T},p^{k_{m}}_{T+1})}{H_{T}}\leq \frac{\delta}{1+\varepsilon}<\delta,
\end{eqnarray*}
where the first equality is derived from Definition \ref{defRealSavings}; the second from $z_{T+1}(p_{T},p_{T+1},p_{T+2})=0$; the third from Walras' law applied to all households from generation $G_{T}$; and the penultimate inequality from Definition \ref{defProneSavings} and (\ref{eqExistSub}). Then, $\sum_{h\in G_{T+1}}s^{h}(p^{k_{m}}_{T+1},p^{k_{m}}_{T+2})/H_{T+1}\leq \delta$ and, by induction, we conclude that (\ref{eqExist3}) remains valid. 

Next, considering the finite economy $\mathcal{F}_{k_{m}}$, we can write
\begin{eqnarray*}
    \sum_{h\in G_{k_{m}}}\frac{s^{h}(p^{k_{m}}_{k_{m}},p^{k_{m}}_{k_{m}+1})}{H_{k_{m}}}&=& \frac{1}{H_{k_{m}}}\sum_{h\in G_{k_{m}}}\frac{p^{k_{m}}_{k_{m}}}{\Vert p^{k_{m}}_{k_{m}}\Vert}\cdot (e^{h}_{k_{m}}-x^{h}_{k_{m}}(p^{k_{m}}_{k_{m}},p^{k_{m}}_{k_{m}+1}))\\
    &=&\frac{1}{H_{k_{m}}\Vert p^{k_{m}}_{k_{m}}\Vert} \sum_{h\in G_{k_{m}}}p^{k_{m}}_{k_{m}+1}\cdot(x^{h}_{k_{m}+1}(p^{k_{m}}_{k_{m}},p^{k_{m}}_{k_{m}+1})-e^{h}_{k_{m}+1})\\
    &=&\frac{H_{k_{m}+1}}{H_{k_{m}}\Vert p^{k_{m}}_{k_{m}}\Vert}p^{k_{m}}_{k_{m}+1}\cdot (e_{\max},\ldots,e_{\max})\\
    &=&\frac{\Vert p^{k_{m}}_{k_{m}+1}\Vert}{\Vert p^{k_{m}}_{k_{m}}\Vert}\alpha_{k_{m}}e_{\max},
\end{eqnarray*}
where the first equality is derived from Definition \ref{defRealSavings}; the second from Walras' law applied to all households from generation $G_{k_{m}}$; and the third from market clearing at $t=k_{m}+1$ and the fact that $x^{*}_{k_{m}+1}(p^{k_{m}}_{k_{m}},p^{k_{m}}_{k_{m}+1})=0$. Then, (\ref{eqExist3}) implies
\begin{eqnarray*}
    \frac{\Vert p^{k_{m}}_{k_{m}}\Vert}{\alpha_{k_{m}}\Vert p^{k_{m}}_{k_{m}+1}\Vert}\geq \frac{e_{\max}}{\delta}\geq1,
\end{eqnarray*}
where the last inequality is due to the fact that we can, without loss of generality, assume $\delta\leq e_{\max}$. However, Definition \ref{defProneSavings} and (\ref{eqExist3}) imply
\begin{eqnarray*}
    \frac{\Vert p^{k_{m}}_{k_{m}}\Vert}{\alpha_{k_{m}}\Vert p^{k_{m}}_{k_{m}+1}\Vert}\leq \frac{1}{1+\varepsilon}<1,
\end{eqnarray*}
absurd. We conclude that $p\in\mathcal{H}^{PO}$ and, therefore, $\mathcal{H}^{PO}$ is non-empty.

Next, let $\{p^{n}\}_{n\geq1}$ be a sequence in $\mathcal{H}^{PO}$ with $\lim_{n\rightarrow\infty}p^{n}=p\in\mathbb{R}^{\infty}$. Since $\mathcal{H}^{PO}\subseteq\mathcal{H}$ and $\mathcal{H}$ is compact by Theorem \ref{theoEqSetCompact}, $p\in\mathcal{H}$. If $p\notin\mathcal{H}^{PO}$, Proposition \ref{propNecAndSuffConditionParetoOpt} implies that $\lim_{t\rightarrow\infty}\sum_{h\in G_{t}}s^{h}(p_{t},p_{t+1})/H_{t}=0$ and, therefore,  there is $T\geq0$ such that 
\begin{eqnarray*}
    \sum_{h\in G_{T}}\frac{s^{h}(p_{T},p_{T+1})}{H_{T}}<\delta/2.
\end{eqnarray*}
Convergence on the product topology implies that there is $N\geq1$ such that 
\begin{eqnarray*}
    \sum_{h\in G_{T}}\frac{s^{h}(p^{N}_{T},p^{N}_{T+1})}{H_{T}}\leq\delta.
\end{eqnarray*}
Then, Lemma \ref{lemmaSecondParetoOptEquilibria} applied to $p^{N}\in\mathcal{H}$ in period $T\geq0$ implies that $p^{N}\notin\mathcal{H}^{PO}$, absurd. Therefore, $p\in\mathcal{H}^{PO}$ and $\mathcal{H}^{PO}$ is closed. We conclude that $\mathcal{H}^{PO}$, as a closed subset of the compact set $\mathcal{H}$, is itself compact.

Lastly, let $\mathcal{S}_{0}\subset\mathcal{H}$ be given by (\ref{eqSj}). Then, $\bigcup_{j\geq0}\mathcal{S}_{j}=\mathcal{H}/\mathcal{H}^{PO}$ implies $\mathcal{H}^{PO}\bigcap \mathcal{S}_{0}=\emptyset$ and, therefore,
\begin{eqnarray*}
    \sum_{h\in G_{0}}\frac{s^{h}(p_{0},p_{1})}{H_{0}}>\delta,
\end{eqnarray*}
for $p\in\mathcal{H}^{PO}$.
\end{proof}

\begin{proof}[Proof of Theorem~{\upshape\ref{theoClassModelExistence}}]
Our goal is to apply Theorem \ref{theoExistenceOfParetoOptJME} and, in order to do so, it is necessary to ``extend'' generation $G^{2}_{0}$ of the economy $\mathcal{E}^{2}$ back to period $t=0$. Let $\zeta\geq1$ and $\mathcal{E}(\zeta)$ be the extended economy we envision, with a single perishable commodity in period $t=0$ (i.e., $L_{0}=1$). Let all generations $G_{t}(\zeta)$, $t\geq1$, from $\mathcal{E}(\zeta)$ coincide with those from $\mathcal{E}^{2}_{\geq1}$. For all households $h\in G^{2}_{0}$, we define $h^{\prime}\in G_{0}(\zeta)$ by a utility function $u^{h^{\prime}}:\mathbb{R}^{1+L_{1}}_{+}\rightarrow\mathbb{R}$, given by
\begin{eqnarray*}
    u^{h^{\prime}}(c_{0},c_{1})=u^{h}(c_{1}),
\end{eqnarray*}
and a nonzero endowment $e^{h^{\prime}}=(e^{h^{\prime}}_{0},e^{h^{\prime}}_{1})=(m^{h},e^{h}_{1})\in\mathbb{R}^{1+L_{1}}_{+}$. Clearly, for all $h^{\prime}\in G_{0}(\zeta)$, $u^{h^{\prime}}(\cdot)$ is continuous, non-decreasing, semi-strictly quasiconcave and without local maxima, and $\sum_{h^{\prime}\in G_{0}(\zeta)}e^{h^{\prime}}\in\mathbb{R}^{1+L_{1}}_{++}$. Utility maximization implies $x^{h^{\prime}}_{0}(p_{0},p_{1})=0$ and, therefore,
\begin{eqnarray*}
    s^{h^{\prime}}(p_{0},p_{1})=e^{h^{\prime}}_{0}=m^{h},
\end{eqnarray*}
for $(p_{0},p_{1})\in\mathbb{R}^{1+L_{1}}_{++}$ and $h^{\prime}\in G_{0}(\zeta)$. In particular, we have
\begin{eqnarray}\label{eqExisClassic1}
    \sum_{h^{\prime}\in G_{0}(\zeta)}s^{h^{\prime}}(p_{0},p_{1})=\sum_{h^{\prime}\in G_{0}(\zeta)}e^{h^{\prime}}_{0}=\sum_{h\in G_{0}^{2}}m^{h}=1.
\end{eqnarray}

Since $\mathcal{E}^{2}_{\geq1}$ satisfies Assumption \ref{assProneSavingsEconomy}, let $\varepsilon,\delta>0$ be given by Definition \ref{defProneSavings}. We can assume, without loss of generality, $\delta<\min\,\{1,(H_{0}+1)^{-1}\}$ and $\Vert e^{1}_{1} \Vert_{\infty}<e_{\max}$ (with this last inequality stemming from the assumptions of the theorem). 

Notice that $\sum_{h^{\prime}\in G_{0}(\zeta)} x^{h^{\prime}}_{0}(p_{0},p_{1})=0$ implies that, with these households only, we will not be able to fulfill Assumption \ref{assBoundsPrices}, since no matter how low $p_{0}=p_{m}\in\mathbb{R}_{++}$ becomes relative to $p_{1}\in\mathbb{R}^{L_{1}}_{++}$, the demand for money will never skyrocket.

Then, we must add a single household $h^{*}_{\zeta}$ to $G_{0}(\zeta)$ (so that the demographic growth rate in period $t=0$ now becomes $\alpha_{0}(\zeta)=H_{1}/(H_{0}+1)$) with a utility function $u^{h^{*}_{\zeta}}:\mathbb{R}^{1+L_{1}}_{+}\rightarrow\mathbb{R}$, given by
\begin{eqnarray*}
    u^{h^{*}_{\zeta}}(c_{0},c_{1})=f(\zeta)\log c_{0}+\sum^{L_{1}}_{i=1}\log c_{1i},
\end{eqnarray*}
with 
\begin{eqnarray*}
    f(\zeta)=\frac{\zeta L_{1}e_{\max}\alpha_{0}(\zeta)}{(e_{\max}-\Vert e^{1}_{1}\Vert_{\infty})(1+\varepsilon)}>0,
\end{eqnarray*}
and an endowment $e^{h^{*}_{\zeta}}=(e_{\max},(e_{\max}-\Vert e^{1}_{1} \Vert_{\infty})/\zeta,\ldots,(e_{\max}-\Vert e^{1}_{1} \Vert_{\infty})/\zeta)\in\mathbb{R}^{1+L_{1}}_{++}$. This star-household provides a strictly positive demand for money
\begin{eqnarray*}
    x^{h^{*}_{\zeta}}(p_{0},p_{1})=\frac{p_{0}e_{\max}+\Vert p_{1}\Vert(e_{\max}-\Vert e^{1}_{1}\Vert_{\infty})/\zeta}{f(\zeta)+L_{1}}\biggr(\frac{f(\zeta)}{p_{0}},\frac{1}{p_{11}},\ldots,\frac{1}{p_{1L_{1}}}\biggr),
\end{eqnarray*}
so that
\begin{eqnarray}\label{eqExisClassic2}
    s^{h^{*}_{\zeta}}(p_{0},p_{1})=e_{\max}-x^{h^{*}_{\zeta}}_{0}(p_{0},p_{1})=\frac{p_{0}L_{1}e_{\max}- f(\zeta)\Vert p_{1}\Vert(e_{\max}-\Vert e^{1}_{1}\Vert_{\infty})/\zeta}{p_{0}(f(\zeta)+L_{1})},
\end{eqnarray}
for $(p_{0},p_{1})\in\mathbb{R}^{1+L_{1}}_{++}$. Notice that
\begin{eqnarray}\label{eqExisClassic3}
    \frac{p_{0}L_{1}e_{\max}-f(\zeta)\Vert p_{1}\Vert(e_{\max}-\Vert e^{1}_{1}\Vert_{\infty})/\zeta}{p_{0}(f(\zeta)+L_{1})}+1\leq \delta \implies
    \frac{p_{0}}{\Vert p_{1}\Vert}\leq \frac{f(\zeta)(e_{\max}-\Vert e^{1}_{1}\Vert_{\infty})/\zeta}{L_{1}e_{\max}+(1-\delta)(f(\zeta)+L_{1})}.
\end{eqnarray}
Since
\begin{eqnarray*}
     \frac{f(\zeta)(e_{\max}-\Vert e^{1}_{1}\Vert_{\infty})/\zeta }{L_{1}e_{\max}+(1-\delta)(f(\zeta)+L_{1})}\leq \frac{f(\zeta)(e_{\max}-\Vert e^{1}_{1}\Vert_{\infty})}{\zeta L_{1}e_{\max}}=\frac{\alpha_{0}(\zeta)}{1+\varepsilon},
\end{eqnarray*}
we conclude, through (\ref{eqExisClassic1}), (\ref{eqExisClassic2}) and (\ref{eqExisClassic3}), that
\begin{eqnarray}\label{eqExisClassic4}
    s^{h^{*}_{\zeta}}(p_{0},p_{1})+\sum_{h^{\prime}\in G_{0}(\zeta)}s^{h^{\prime}}(p_{0},p_{1})\leq \delta\implies\frac{p_{0}}{\Vert p_{1}\Vert}\leq \frac{\alpha_{0}(\zeta)}{1+\varepsilon},
\end{eqnarray}
for $(p_{0},p_{1})\in\mathbb{R}^{1+L_{1}}_{++}$.

We proceed to show that $\mathcal{E(\zeta)}$ satisfies Assumptions \ref{assBoundsPopAndEndownments}--\ref{assBoundsPrices} and \ref{assProneSavingsEconomy}.

First, since all generations $G_{t}(\zeta)$, $t\geq1$, from the extended economy $\mathcal{E}(\zeta)$ coincide with those from $\mathcal{E}^{2}_{\geq1}$, and $\mathcal{E}^{2}_{\geq1}$ satisfies Assumption \ref{assBoundsPopAndEndownments}, we only need to check generation $G_{0}(\zeta)$. The demographic growth rate is
\begin{eqnarray*}
    \alpha_{0}(\zeta)=\frac{H_{1}}{H_{0}+1},
\end{eqnarray*}
since $G_{0}(\zeta)$ differs from $G^{2}_{0}$ only by the ``star-household.'' By assumption of the theorem, we have
\begin{eqnarray*}
    \alpha_{\min}\leq\frac{H_{1}}{H_{0}+1}\leq\alpha_{\max},
\end{eqnarray*}
and, therefore, $\alpha_{\min}\leq \alpha_{0}(\zeta)\leq \alpha_{\max}$. Since $m^{h}\leq e_{\max}$, $h\in G_{0}^{2}$, we have $\Vert e^{h^{\prime}}\Vert_{\infty}\leq e_{\max}$, $h^{\prime}\in G_{0}(\zeta)$. Also, $\zeta\geq1$ implies $\Vert e^{h^{*}_{\zeta}}\Vert_{\infty}\leq e_{\max}$. We conclude that $\mathcal{E}(\zeta)$ satisfies Assumption \ref{assBoundsPopAndEndownments}. Assumption \ref{assResourceRelated} is also satisfied by $\mathcal{E}(\zeta)$, since $h^{*}_{\zeta}$ is resource related to at least one $h^{\prime}\in G_{0}(\zeta)$ and households from $\mathcal{E}^{2}$ are indirectly resource related.

Let $\sigma_{t}\in(0,1)$, $t\geq1$, be given by Assumption \ref{assBoundsPrices} when considering the economy $\mathcal{E}^{2}_{\geq1}$, so that 
\begin{eqnarray}\label{eqExisClassic5}
\biggr(\frac{p_{t}}{p_{t1}},\frac{p_{t+1}}{p_{t1}}\biggr)\notin\mathcal{B}_{t}(\sigma_{t})\implies \biggr\Vert \sum_{h\in G_{t}(\zeta)}\frac{x^{h}(p_{t},p_{t+1})}{H_{t}}\biggr\Vert=\biggr\Vert \sum_{h\in G^{2}_{t}}\frac{x^{h}(p_{t},p_{t+1})}{H_{t}}\biggr\Vert>\beta e_{\max},
\end{eqnarray}
for $(p_{t},p_{t+1})\in\mathbb{R}^{L_{t}+L_{t+1}}_{++}$, $t\geq1$. 

Next, notice that there is $\sigma_{0}(\zeta)\in(0,1)$ such that
\begin{eqnarray*}
    \biggr(1,\frac{p_{1}}{p_{0}}\biggr)\notin\mathcal{B}_{0}(\sigma_{0}(\zeta))\implies \biggr\Vert \frac{x^{h^{*}_{\zeta}}(p_{0},p_{1})}{H_{0}+1}\biggr\Vert_{\infty}>\beta e_{\max},
\end{eqnarray*}
for $(p_{0},p_{1})\in\mathbb{R}^{1+L_{1}}_{++}$. Since
\begin{eqnarray*}
    \biggr\Vert x^{h^{*}_{\zeta}}(p_{0},p_{1})+ \sum_{h^{\prime}\in G_{0}(\zeta)}x^{h^{\prime}}(p_{0},p_{1})\biggr\Vert_{\infty}>\Vert x^{h^{*}_{\zeta}}(p_{0},p_{1})\Vert_{\infty},
\end{eqnarray*}
we have
\begin{eqnarray}\label{eqExisClassic6}
    \biggr(1,\frac{p_{1}}{p_{0}}\biggr)\notin\mathcal{B}_{0}(\sigma_{0}(\zeta))\implies \biggr\Vert \frac{x^{h^{*}_{\zeta}}(p_{0},p_{1})+ \sum_{h^{\prime}\in G_{0}(\zeta)}x^{h^{\prime}}(p_{0},p_{1})}{H_{0}+1}\biggr\Vert_{\infty}>\beta e_{\max},
\end{eqnarray}
for $(p_{0},p_{1})\in\mathbb{R}^{1+L_{1}}_{++}$. We conclude, by (\ref{eqExisClassic5}) and (\ref{eqExisClassic6}), that $\mathcal{E}(\zeta)$ satisfies Assumption \ref{assBoundsPrices}.

It remains to show that Assumption \ref{assProneSavingsEconomy} is satisfied. Since all generations $G_{t}(\zeta)$, $t\geq1$, from $\mathcal{E}(\zeta)$ coincide with those from $\mathcal{E}^{2}_{\geq1}$, and $\mathcal{E}^{2}_{\geq1}$ satisfies Assumption \ref{assProneSavingsEconomy}, we only need, once again, to check generation $G_{0}(\zeta)$. For $(1,p_{1}/p_{0})\in \mathcal{B}_{0}(\sigma_{0}(\zeta))$, $(p_{0},p_{1})\in\mathbb{R}^{1+L_{1}}_{++}$, the result follows directly from (\ref{eqExisClassic4}).

Lastly, notice that we do not need to verify whether Assumption \ref{assCassCriterion} is satisfied by $\mathcal{E}(\zeta)$. If it is not, then the subset $\mathcal{H}^{PO}(\zeta)$ must be read not as the ``subset of Pareto optimal equilibria,'' but as the ``subset of equilibria that satisfy the Cass criterion.'' 

This possible different interpretation, however, does not prevent us from applying Theorem \ref{theoExistenceOfParetoOptJME} to $\mathcal{E}(n)$, $n\geq1$, to find a sequence $\{p^{n}\}_{n\geq1}$ such that $p^{n}\in\mathcal{H}^{PO}(n)$. Then, the market clearing equation in period $t=1$ allows us to write
\begin{eqnarray*}
    \biggr\Vert \sum_{h\in G^{2}_{0}}\frac{x^{h}_{1}(1,p^{n}_{1})}{H_{0}}\biggr\Vert_{\infty} &=& \biggr\Vert \sum_{h^{\prime}\in G_{0}(n)}\frac{x^{h^{\prime}}_{1}(p^{n}_{0},p^{n}_{1})}{H_{0}}\biggr\Vert_{\infty}\\
    &\leq&\frac{1}{H_{0}}\biggr\Vert e^{h^{*}_{n}}_{1}+\sum_{h^{\prime}\in G_{0}(n)}e^{h^{\prime}}_{1}+\sum_{h\in G_{1}(n)}e^{h}_{1}\biggr\Vert_{\infty}\\
    &=&\frac{1}{H_{0}}\biggr\Vert e^{h^{*}_{n}}_{1}+\sum_{h\in G^{2}_{0}} e^{h}_{1}+\sum_{h\in G^{2}_{1}} e^{h}_{1}\biggr\Vert_{\infty}\\
    &\leq&\frac{1}{H_{0}}\biggr(\Vert e^{h^{*}_{n}}_{1}+e^{1}_{1}\Vert_{\infty}+(H_{0}-1)e_{\max}+H_{1}e_{\max}\biggr)\\
    &\leq&\frac{(H_{0}+H_{1})}{H_{0}}e_{\max}\\
    &\leq&\beta e_{\max}.
\end{eqnarray*}
By assumption of the theorem, this implies $p^{n}_{11}>\lambda>0$, $n\geq1$, and therefore 
\begin{eqnarray}\label{eqBoundPn11}
   0<\frac{1}{p^{n}_{11}}<\frac{1}{\lambda},
\end{eqnarray}
for $n\geq1$. Next, let $q^{n}=(q^{n}_{1},q^{n}_{2},\ldots)=(p^{n}_{1}/p^{n}_{11},p^{n}_{2}/p^{n}_{11},\ldots)\in\mathbb{R}^{\infty}_{++}$, $n\geq1$, so that $q^{n}_{11}=1$, $p^{n}=(1,p^{n}_{11}q^{n})$, and
\begin{eqnarray}\label{eqExisClassic7}
    \sum^{\infty}_{t=1}\frac{1}{H_{t}\Vert q^{n}_{t}\Vert}=p^{n}_{11}\sum^{\infty}_{t=1}\frac{1}{H_{t}\Vert p^{n}_{t}\Vert}=+\infty,
\end{eqnarray}
for $n\geq1$. Notice that the market clearing equation in period $t=1$ for the economy $\mathcal{E}(n)$ and the homogeneity of demand imply
\begin{eqnarray}
    \biggr\Vert \sum_{h\in G^{2}_{1}} \frac{x^{h}_{1}(q^{n}_{1},q^{n}_{2})}{H_{1}}\biggr\Vert_{\infty}&=&\biggr\Vert \sum_{h\in G_{1}(n)} \frac{x^{h}_{1}(q^{n}_{1},q^{n}_{2})}{H_{1}}\biggr\Vert_{\infty}\nonumber\\
    &\leq& \frac{1}{H_{1}}\biggr\Vert x^{h^{*}_{n}}_{1}(p^{n}_{0},p^{n}_{1})+\sum_{h^{\prime}\in G_{0}(n)} x^{h^{\prime}}_{1}(p^{n}_{0},p^{n}_{1})+\sum_{h\in G_{1}(n)} x^{h}_{1}(p^{n}_{1},p^{n}_{2})\biggr\Vert_{\infty}\nonumber\\
    &=&\frac{1}{H_{1}}\biggr\Vert e^{h^{*}_{n}}_{1}+\sum_{h\in G^{2}_{0}} e^{h}_{1}+\sum_{h\in G^{2}_{1}} e^{h}_{1}\biggr\Vert_{\infty}\nonumber\\
    &\leq&\frac{1}{H_{1}}\biggr(\Vert e^{h^{*}_{n}}_{1}+e^{1}_{1}\Vert_{\infty}+(H_{0}-1)e_{\max}+H_{1}e_{\max}\biggr)\nonumber\\
    &\leq&\frac{(H_{0}+H_{1})}{H_{1}}e_{\max}\nonumber\\
    &\leq&\beta e_{\max}. \label{eqFirstBoundG1}
\end{eqnarray}
Furthermore, the market clearing equation in period $t\geq2$ and the homogeneity of demand imply
\begin{eqnarray}
    \sum_{h\in G^{2}_{t-1}}x^{h}_{t}(q^{n}_{t-1},q^{n}_{t})+\sum_{h\in G^{2}_{t}}x^{h}_{t}(q^{n}_{t},q^{n}_{t+1})&=&\sum_{h\in G_{t-1}(n)}x^{h}_{t}(p^{n}_{t-1},p^{n}_{t})+\sum_{h\in G_{t}(n)}x^{h}_{t}(p^{n}_{t},p^{n}_{t+1})\nonumber\\
    &=&\sum_{h\in G_{t-1}(n)}e^{h}_{t}+\sum_{h\in G_{t}(n)}e^{h}_{t}\nonumber\\
    &=&\sum_{h\in G^{2}_{t-1}}e^{h}_{t}+\sum_{h\in G^{2}_{t}}e^{h}_{t}.\label{eqExisClassic8}
\end{eqnarray}
In particular, for $t=2$ we have
\begin{eqnarray}\label{eqSecondBoundG1}
    \biggr\Vert \sum_{h\in G^{2}_{1}}\frac{x^{h}_{2}(q^{n}_{1},q^{n}_{2})}{H_{1}}\biggr\Vert_{\infty}\leq \beta e_{\max}.
\end{eqnarray}
Therefore, (\ref{eqFirstBoundG1}) and (\ref{eqSecondBoundG1}) imply
\begin{eqnarray*}
    \biggr\Vert \sum_{h\in G^{2}_{1}}\frac{x^{h}(q^{n}_{1},q^{n}_{2})}{H_{1}}\biggr\Vert_{\infty}\leq \beta e_{\max},
\end{eqnarray*}
and, by Assumption \ref{assBoundsPrices}, we have 
\begin{eqnarray}\label{eqExisClassic9}
    (q^{n}_{1},q^{n}_{2})\in\mathcal{B}_{1}(\sigma_{1}),
\end{eqnarray}
for $n\geq1$.

Let $\mathcal{H}^{PO}_{\geq1}$ be the subset of Pareto optimal equilibria from $\mathcal{E}^{2}_{\geq1}$. Notice that (\ref{eqExisClassic7}), (\ref{eqExisClassic8}), and (\ref{eqExisClassic9}) imply $q^{n}\in\mathcal{H}^{PO}_{\geq1}$ (there is a one-period relabeling of time in this assertion since $\mathcal{E}^{2}_{\geq1}$ ``starts'' with generation $G^{2}_{1}$). Theorem \ref{theoExistenceOfParetoOptJME} states that $\mathcal{H}^{PO}_{\geq1}$ is compact. Furthermore, (\ref{eqBoundPn11}) implies that $\{1/p^{n}_{11}\}_{n\geq1}$ is bounded. Therefore, we can assume, passing to a subsequence if necessary, that
\begin{eqnarray}\label{eqLimQn}
\lim_{n\rightarrow\infty}q^{n}=q\in\mathcal{H}^{PO}_{\geq1}\subset{R}^{\infty}_{++},
\end{eqnarray}
and that
\begin{eqnarray}\label{eqLimInversePn11}
    \lim_{n\rightarrow\infty} \frac{1}{p^{n}_{11}}=\mu\geq0.
\end{eqnarray}
In particular, (\ref{eqLimQn}) implies
\begin{eqnarray}\label{eqLimQn2}
    \sum^{\infty}_{t=1}\frac{1}{H_{t}\Vert q_{t}\Vert}=+\infty.
\end{eqnarray}
Since $p^{n}\in\mathcal{H}^{PO}(n)$, $n\geq1$, the last part of Theorem \ref{theoExistenceOfParetoOptJME} implies
\begin{eqnarray*}
    s^{h^{*}_{n}}(p^{n}_{0},p^{n}_{1})+\sum_{h^{\prime}\in G_{0}(n)}s^{h^{\prime}}(p^{n}_{0},p^{n}_{1})=s^{h^{*}_{n}}(1,p^{n}_{11}q^{n}_{1})+1>(H_{0}+1)\delta,
\end{eqnarray*}
so that
\begin{eqnarray*}
    x^{h^{*}_{n}}_{0}(1,p^{n}_{11}q^{n}_{1})=\frac{e_{\max}+p^{n}_{11}\Vert q^{n}_{1}\Vert(e_{\max}-\Vert e^{1}_{1}\Vert_{\infty})/n}{1+L_{1}/f(n)}<e_{\max}+1-(H_{0}+1)\delta.
\end{eqnarray*}
Rearranging terms, we obtain
\begin{eqnarray}\label{eqBoundsPn11}
    0<\frac{p^{n}_{11}}{n}<\frac{(1+L_{1}/f(n))(e_{\max}+1-(H_{0}+1)\delta)-e_{\max}}{\Vert q^{n}_{1}\Vert(e_{\max}-\Vert e^{1}_{1}\Vert_{\infty})}.
\end{eqnarray}
Since $\lim_{n\rightarrow\infty}\Vert q^{n}_{1}\Vert=\Vert q_{1}\Vert>0$, $1-(H_{0}+1)\delta>0$ and
\begin{eqnarray*}
    \lim_{n\rightarrow\infty} f(n)=\lim_{n\rightarrow\infty}\frac{nL_{1}e_{\max}H_{1}}{(e_{\max}-\Vert e^{1}_{1}\Vert_{\infty})(1+\varepsilon)(H_{0}+1)}=+\infty,
\end{eqnarray*}
(\ref{eqBoundsPn11}) implies that $\{p^{n}_{11}/n\}_{n\geq1}$ is bounded. 
Then, (\ref{eqBoundPn11}) implies
\begin{eqnarray}
    \lim_{n\rightarrow\infty}x^{h^{*}_{n}}_{1}(p^{n}_{0},p^{n}_{1})&=&\lim_{n\rightarrow\infty}x^{h^{*}_{n}}_{1}(1,p^{n}_{11}q^{n}_{1})\nonumber\\
    &=&\lim_{n\rightarrow\infty}\frac{e_{\max}+p^{n}_{11}\Vert q^{n}_{1}\Vert(e_{\max}-\Vert e^{1}_{1}\Vert_{\infty})/n}{f(n)+L_{1}}\frac{1}{p^{n}_{11}}\biggr(\frac{1}{q^{n}_{11}},\ldots,\frac{1}{q^{n}_{1L_{1}}}\biggr)\nonumber\\
    &< &\frac{1}{\lambda}\lim_{n\rightarrow\infty}\frac{e_{\max}+p^{n}_{11}\Vert q^{n}_{1}\Vert(e_{\max}-\Vert e^{1}_{1}\Vert_{\infty})/n}{f(n)+L_{1}}\biggr(\frac{1}{q^{n}_{11}},\ldots,\frac{1}{q^{n}_{1L_{1}}}\biggr)\nonumber\\
    &=&0\label{eqHstarVanish}.
\end{eqnarray}

Next, since $q\in \mathcal{H}^{PO}_{\geq1}$, we have
\begin{eqnarray}\label{eqMarketClearingLimitq}
     \sum_{h\in G^{2}_{t-1}}x^{h}_{t}(q_{t-1},q_{t})+\sum_{h\in G^{2}_{t}}x^{h}_{t}(q_{t},q_{t+1})=\sum_{h\in G^{2}_{t-1}}e^{h}_{t}+\sum_{h\in G^{2}_{t}}e^{h}_{t},
\end{eqnarray}
for $t\geq2$. The market clearing equation for $\mathcal{E}(n)$ at $t=1$ is given by
\begin{eqnarray*}
    e^{h^{*}_{n}}_{1}+\sum_{h\in G^{2}_{0}} e^{h}_{1}+\sum_{h\in G^{2}_{1}} e^{h}_{1}
    &=&x^{h^{*}_{n}}_{1}(p^{n}_{0},p^{n}_{1})+\sum_{h^{\prime}\in G_{0}(n)} x^{h^{\prime}}_{1}(p^{n}_{0},p^{n}_{1})+\sum_{h\in G^{2}_{1}} x^{h}_{1}(p^{n}_{1},p^{n}_{2})\\
    &=&x^{h^{*}_{n}}_{1}(1,p^{n}_{1})+\sum_{h\in G^{2}_{0}} x^{h}_{1}(1/p^{n}_{11},q^{n}_{1})+\sum_{h\in G^{2}_{1}} x^{h}_{1}(q^{n}_{1},q^{n}_{2}),
\end{eqnarray*}
for $n\geq1$. Then, (\ref{eqHstarVanish}) and $\lim_{n\rightarrow+\infty}e^{h^{*}_{n}}_{1}=0$ imply
\begin{eqnarray}\label{eqMarketClearingLimitT1}
    \sum_{h\in G^{2}_{0}} x^{h}_{1}(\mu,q_{1})+\sum_{h\in G^{2}_{1}} x^{h}_{1}(q_{1},q_{2})=\sum_{h\in G^{2}_{0}} e^{h}_{1}+\sum_{h\in G^{2}_{1}} e^{h}_{1}.
\end{eqnarray}
Lemma \ref{lemmaSecondParetoOptEquilibria} and $p^{n}\in\mathcal{H}^{PO}(n)$, $n\geq1$, imply that real savings per capita of generation $G_{1}(n)$ from the economy $\mathcal{E}(n)$ has $\delta>0$ as lower bound. Therefore, 
\begin{eqnarray*}
    \delta&\leq& \sum_{h\in G_{1}(n)}\frac{s^{h}(p^{n}_{1},p^{n}_{2})}{H_{1}}\\
    &=&\frac{p^{n}_{1} }{H_{1}\Vert p^{n}_{1}\Vert}\cdot \sum_{h\in G^{2}_{1}}(e^{h}_{1}-x^{h}_{1}(p^{n}_{1},p^{n}_{2}))\\
    &=&\frac{q^{n}_{1} }{H_{1}\Vert q^{n}_{1}\Vert}\cdot\biggr(x^{h^{*}_{n}}_{1}(p^{n}_{0},p^{n}_{1})-e^{h^{*}_{n}}_{1}+\sum_{h\in G^{2}_{0}} (x^{h}_{1}(1/p^{n}_{11},q^{n}_{1})- e^{h}_{1})\biggr),
\end{eqnarray*}
where the last equality is due to the market clearing equation in period $t=1$. Letting $n\rightarrow\infty$, we obtain
\begin{eqnarray*}
    \sum_{h\in G^{2}_{0}} \frac{q_{1}\cdot(x^{h}_{1}(\mu,q_{1})- e^{h}_{1})}{{H_{1}\Vert q_{1}\Vert}}\geq\delta.
\end{eqnarray*}
Then, Walras' law implies
\begin{eqnarray*}
    \sum_{h\in G^{2}_{0}} \frac{q_{1}\cdot(x^{h}_{1}(\mu, q_{1})- e^{h}_{1})}{H_{1}\Vert q_{1}\Vert}=\sum_{h\in G^{2}_{0}} \frac{\mu m^{h}}{H_{1}\Vert q_{1}\Vert}=\frac{\mu}{H_{1}\Vert q_{1}\Vert}\geq\delta>0,
\end{eqnarray*}
and, therefore, $\mu>0$. Since $\mathcal{E}^{2}$ satisfies Assumption \ref{assCassCriterion}, (\ref{eqLimQn2}), (\ref{eqMarketClearingLimitq}) and (\ref{eqMarketClearingLimitT1}) allow us to conclude that $(\mu,q)\in\mathbb{R}^{\infty}_{++}$ is an optimal monetary equilibrium of $\mathcal{E}^{2}$.
\end{proof}

\begin{proof}[Proof of Lemma~{\upshape\ref{lemmaTailEquilibria}}]
Directly derived from Definitions \ref{defSetEquilibria}, \ref{defParetoOptimalEqSet} and \ref{defFiniteRep}.
\end{proof}

\begin{proof}[Proof of Lemma~{\upshape\ref{lemmaFiniteReplEqSet}}]
By Definition \ref{defFiniteRep}, all generations of $\mathcal{E}_{k}$ coincide with those from $\mathcal{E}$ until period $k\geq1$. Therefore, the set of $j$-sighted equilibria $\mathcal{V}_{j}$ of $\mathcal{E}$ coincides with the set of $j$-sighted equilibria $\mathcal{V}_{kj}$ of $\mathcal{E}_{k}$, for $1\leq j\leq k$. 

Theorem \ref{theoEqSetCompact} implies that $\mathcal{H}_{k}\subseteq \mathcal{V}^{\infty}_{kj}$, $j,k\geq1$. In particular, we have
\begin{eqnarray}\label{eqInclusion}
    \mathcal{H}_{k}\subseteq \mathcal{V}^{\infty}_{kj}=\mathcal{V}^{\infty}_{j},
\end{eqnarray}
for $k\geq j$. Then, notice that $\{\overline{\bigcup_{k\geq j}\mathcal{H}_{k}}\}_{j\geq1}$ and $\{\mathcal{V}^{\infty}_{j}\}_{j\geq1}$ are nonincreasing sequences of sets with well-defined limits and (\ref{eqInclusion}) allows us to write
\begin{eqnarray*}
\bigcup_{k\geq j}\mathcal{H}_{k}\subseteq \mathcal{V}^{\infty}_{j}\implies \overline{\bigcup_{k\geq j}\mathcal{H}_{k}}\subseteq \overline{\mathcal{V}^{\infty}_{j}}=\mathcal{V}^{\infty}_{j}\implies\lim_{j\rightarrow\infty}\overline{\bigcup_{k\geq j}\mathcal{H}_{k}}\subseteq \lim_{j\rightarrow\infty}\mathcal{V}^{\infty}_{j}=\mathcal{H},
\end{eqnarray*}
where $\overline{\mathcal{V}^{\infty}_{j}}=\mathcal{V}^{\infty}_{j}$, $j\geq1$, due to Lemma \ref{lemmaVjInfinite} and $\lim_{j\rightarrow\infty}\mathcal{V}^{\infty}_{j}=\mathcal{H}$ due to Theorem \ref{theoEqSetCompact}. Theorem \ref{theoEqSetCompact} also implies $\mathcal{H}_{k}\neq\emptyset$, $k\geq1$. From Definition \ref{defFiniteRep}, there is $\mathcal{K}^{\prime}\subset\mathbb{R}^{\infty}$ such that $\mathcal{H}_{k}\subseteq\mathcal{K}^{\prime}$, $k\geq1$, and we can write
\begin{eqnarray*}
\bigcup_{k\geq j}\mathcal{H}_{k}\subseteq \mathcal{K}^{\prime}\implies \overline{\bigcup_{k\geq j}\mathcal{H}_{k}}\subseteq \overline{\mathcal{K}^{\prime}}=\mathcal{K}^{\prime},
\end{eqnarray*}
for $j\geq1$. As a closed subset of a compact set, $\overline{\bigcup_{k\geq j}\mathcal{H}_{k}}$, $j\geq1$, is itself compact. Therefore, $\{\overline{\bigcup_{k\geq j}\mathcal{H}_{k}}\}_{j\geq1}$ is a nonincreasing sequence of non-empty compact sets. By Cantor's Intersection Theorem, we have
\begin{eqnarray*}
\lim_{j\rightarrow\infty}\overline{\bigcup_{k\geq j}\mathcal{H}_{k}}\neq\emptyset.    
\end{eqnarray*}
Furthermore, as the intersection of compact sets, $\lim_{j\rightarrow\infty}\overline{\bigcup_{k\geq j}\mathcal{H}_{k}}$ is itself compact. We conclude that $\lim_{j\rightarrow\infty}\overline{\bigcup_{k\geq j}\mathcal{H}_{k}}$ is a non-empty compact set.
\end{proof}

\begin{proof}[Proof of Theorem~{\upshape\ref{theoAlgorithmFinal}}]
Since $\mathcal{H}^{PO}_{k}\subseteq\mathcal{H}_{k}$, $k\geq1$, then
\begin{eqnarray*}
 \bigcup_{k\geq j}\mathcal{H}^{PO}_{k}\subseteq \bigcup_{k\geq j}\mathcal{H}_{k} &\implies& \overline{\bigcup_{k\geq j}\mathcal{H}^{PO}_{k}}\subseteq \overline{\bigcup_{k\geq j}\mathcal{H}_{k}}\\
&\implies&\lim_{j\rightarrow\infty}\overline{\bigcup_{k\geq j}\mathcal{H}^{PO}_{k}}\subseteq \lim_{j\rightarrow\infty}\overline{\bigcup_{k\geq j}\mathcal{H}_{k}}\subseteq\mathcal{H},
\end{eqnarray*}
where the last inclusion is due to Lemma \ref{lemmaFiniteReplEqSet}. 

Next, let $p\in \lim_{j\rightarrow\infty}\overline{\bigcup_{k\geq j}\mathcal{H}^{PO}_{k}}\subseteq\mathcal{H}$. By definition of the nonincreasing set sequence $\{\overline{\bigcup_{k\geq j}\mathcal{H}^{PO}_{k}}\}_{j\geq1}$, we have, for all $j\geq1$, sequences $\{p^{jn}\}_{n\geq1}$ such that $p^{jn}\in \bigcup_{k\geq j}\mathcal{H}^{PO}_{k}$, $n\geq 1$, and $\lim_{n\rightarrow\infty}p^{jn}=p$. 

Let $j_{1}=1$ and notice that the convergence $\lim_{n\rightarrow\infty}p^{j_{1}n}=p$ in the product topology implies that there is $n_{1}\geq1$ such that $\Vert p^{j_{1}n_{1}}_{0}-p_{0}\Vert \leq 1$. Since $p^{j_{1}n_{1}}\in \bigcup_{k\geq j_{1}}\mathcal{H}^{PO}_{k}$, there is $k_{1}\geq j_{1}$ such that $p^{j_{1}n_{1}}\in \mathcal{H}^{PO}_{k_{1}}$. Then, for $j_{2}=k_{1}+1$, convergence $\lim_{n\rightarrow\infty}p^{j_{2}n}=p$ on the product topology implies the existence of $n_{2}\geq1$ such that
\begin{eqnarray*}
    \Vert p^{j_{2}n_{2}}_{0}-p_{0}\Vert &\leq& 1/2\\
    \Vert p^{j_{2}n_{2}}_{1}-p_{1}\Vert &\leq& 1/2.
\end{eqnarray*}

Since $p^{j_{2}n_{2}}\in \bigcup_{k\geq j_{2}}\mathcal{H}^{PO}_{k}$, there is $k_{2}\geq j_{2}>k_{1}$ such that $p^{j_{2}n_{2}}\in \mathcal{H}^{PO}_{k_{2}}$. Proceed with this algorithm in order to define a sequence $\{p^{m}\}_{m\geq1}$, with $p^{m}=p^{j_{m}n_{m}}$, $m\geq1$, and an increasing sequence $\{k_{m}\}_{m\geq1}$, with $\lim_{m\rightarrow\infty}k_{m}=+\infty$, such that $p^{m}\in \mathcal{H}^{PO}_{k_{m}}$ and $\Vert p^{m}_{t}-p_{t}\Vert \leq 1/m$, for $t<m$. This last inequality implies that $\lim_{m\rightarrow\infty}p^{m}_{t}=p_{t}$, $t\geq0$, and, therefore, $\lim_{m\rightarrow\infty}p^{m}=p$.

Suppose $p\notin\mathcal{H}^{PO}$. Proposition \ref{propNecAndSuffConditionParetoOpt} implies that $\lim_{t\rightarrow\infty}\sum_{h\in G_{t}}s^{h}(p_{t},p_{t+1})/H_{t}=0$. Therefore, there is $T\geq0$ such that $\sum_{h\in G_{T}}s^{h}(p_{T},p_{T+1})/H_{T}<\delta/2$. Since $\{p^{m}\}_{m\geq1}$ converges towards $p$, there is $M_{1}\geq1$ such that 
\begin{eqnarray*}
    \sum_{h\in G_{T}}\frac{s^{h}(p^{m}_{T},p^{m}_{T+1})}{H_{T}}\leq\delta,
\end{eqnarray*}
for $m\geq M_{1}$. Also, $\lim_{m\rightarrow\infty}k_{m}=+\infty$ implies that there is $M_{2}\geq1$ such that $k_{M_{2}}>T$. Let $M=\max\{M_{1},M_{2}\}$. Since the sequence $\{k_{m}\}_{m\geq1}$ is increasing, then $k_{M}\geq k_{M_{2}}>T$. 

Economy $\mathcal{E}_{k_{M}}$ has all generations until period $k_{M}>T$ identical to the ones from $\mathcal{E}$ and, in period $t=T$, particularly, the definition of $M_{1}$ implies the following inequality
\begin{eqnarray*}
    \sum_{h\in G_{T}}\frac{s^{h}(p^{M}_{T},p^{M}_{T+1})}{H_{T}}\leq\delta.
\end{eqnarray*}

Since $p^{M}\in \mathcal{H}_{k_{M}}$, Lemma \ref{lemmaSecondParetoOptEquilibria} implies that $p^{M}\notin \mathcal{H}^{PO}_{k_{M}}$, absurd. We conclude that $p\in\mathcal{H}^{PO}$ and, therefore, 
\begin{eqnarray*}
    \lim_{j\rightarrow\infty}\overline{\bigcup_{k\geq j}\mathcal{H}^{PO}_{k}}\subseteq \mathcal{H}^{PO}.
\end{eqnarray*}
Theorem \ref{theoExistenceOfParetoOptJME} implies $\mathcal{H}^{PO}_{k}\neq\emptyset$, $k\geq1$. From Definition \ref{defFiniteRep}, there is $\mathcal{K}^{\prime}\subset\mathbb{R}^{\infty}$ such that $\mathcal{H}_{k}\subseteq\mathcal{K}^{\prime}$, $k\geq1$, and we can write
\begin{eqnarray*}
\bigcup_{k\geq j}\mathcal{H}^{PO}_{k}\subseteq\bigcup_{k\geq j}\mathcal{H}_{k}\subseteq \mathcal{K}^{\prime}\implies \overline{\bigcup_{k\geq j}\mathcal{H}^{PO}_{k}}\subseteq \overline{\mathcal{K}^{\prime}}=\mathcal{K}^{\prime},
\end{eqnarray*}
for $j\geq1$. As a closed subset of a compact set, $\overline{\bigcup_{k\geq j}\mathcal{H}^{PO}_{k}}$, $j\geq1$, is itself compact. Therefore, $\{\overline{\bigcup_{k\geq j}\mathcal{H}^{PO}_{k}}\}_{j\geq1}$ is a nonincreasing sequence of non-empty compact sets. By Cantor's Intersection Theorem, we have 
\begin{eqnarray*}
    \lim_{j\rightarrow\infty}\overline{\bigcup_{k\geq j}\mathcal{H}^{PO}_{k}}\neq\emptyset.
\end{eqnarray*}
Furthermore, as the intersection of compact sets, $\lim_{j\rightarrow\infty}\overline{\bigcup_{k\geq j}\mathcal{H}^{PO}_{k}}$ is itself compact. We conclude that $\lim_{j\rightarrow\infty}\overline{\bigcup_{k\geq j}\mathcal{H}^{PO}_{k}}$ is a non-empty compact set.
\end{proof}

\printbibliography
\end{document}